%% file: main.tex
\documentclass[11pt]{article}

\usepackage{fullpage}

\usepackage{amsmath}
\usepackage{amssymb}
\usepackage{amsthm}
\usepackage{thmtools}

\usepackage{combelow}

\usepackage{xcolor}
\usepackage{xspace}
\usepackage[shortlabels]{enumitem}
\usepackage{tikz}

\usepackage{graphicx}
\usepackage{caption}

\usepackage[
style=alphabetic,
backref=true,
doi=false,
url=false,
maxcitenames=3,
mincitenames=3,
maxnames=10,
minbibnames=10,
backend=bibtex8,
sortlocale=en_US
]{biblatex}

\usepackage[ocgcolorlinks]{hyperref} 
\usepackage{cleveref}
\allowdisplaybreaks

\input{commands} 
\renewcommand{\paragraph}[1]{\medskip\noindent{\bf #1}\xspace}

\colorlet{DarkRed}{red!50!black}
\colorlet{DarkGreen}{green!50!black}
\colorlet{DarkBlue}{blue!50!black}

\hypersetup{
	linkcolor = DarkRed,
	citecolor = DarkGreen,
	urlcolor = DarkBlue,
	bookmarks = true,
	bookmarksnumbered = true,
	linktocpage = true
}

\declaretheorem[numberwithin=section]{theorem}
\declaretheorem[numberlike=theorem]{lemma}
\declaretheorem[numberlike=theorem]{proposition}
\declaretheorem[numberlike=theorem]{corollary}
\declaretheorem[numberlike=theorem]{definition}
\declaretheorem[numberlike=theorem]{claim}

\declaretheorem[numberlike=theorem]{conjecture}

\crefname{algorithm}{Algorithm}{Algorithms}
\Crefname{algorithm}{Algorithm}{Algorithms}

\newcommand{\cX}{\mathcal{X}}
\newcommand{\cI}{\mathcal{I}}
\newcommand{\wtGamma}{\widetilde{\Gamma}}
\newcommand{\wtU}{\widetilde{U}}
\newcommand{\wtPi}{\widetilde{\Pi}}
\newcommand{\wtG}{\widetilde{G}}
\newcommand{\wtV}{\widetilde{V}}
\newcommand{\wtE}{\widetilde{E}}
\usepackage{framed}

\begin{document}

\title{From Gap-ETH to FPT-Inapproximability:\\ Clique, Dominating Set, and More} 
\author{
	Parinya Chalermsook\thanks{Aalto University, Finland. Email: \texttt{parinya.chalermsook@aalto.fi}.}
	\and
	Marek Cygan\thanks{Institute of Informatics, University of Warsaw, Poland. 
	Email: \texttt{cygan@mimuw.edu.pl}.
	} 
	\and
	Guy Kortsarz\thanks{Rutgers University-Camden, New Jersey, USA.
		Email:\texttt{guyk@camden.rutgers.edu}.
        }
	\and
        Bundit Laekhanukit\thanks{
        Weizmann Institute of Science, Israel \&
        Shanghai University of Finance and Economics.
        Email: \texttt{bundit.laekhanukit@weizmann.ac.il}
        }
	\and
	Pasin Manurangsi\thanks{University of California, Berkeley, USA.
	Email: \texttt{pasin@berkeley.edu}.}
	\and
	Danupon Nanongkai\thanks{KTH Royal Institute of Technology, Sweden. Email: \texttt{danupon@kth.se}
	}	
	\and
	Luca Trevisan\thanks{University of California, Berkeley, USA.
	Email: \texttt{luca@berkeley.edu}.}
}

\date{\today}

\hypersetup{
	pdftitle = {FPT-Inapproximability for Clique, Dominating Set, and More},
	pdfauthor = {}
}

%
%

\begin{titlepage}
	\maketitle
	\pagenumbering{roman}
	\input{abstract}

	\newpage
	\setcounter{tocdepth}{2}
	\tableofcontents
\end{titlepage}

\newpage
\pagenumbering{arabic}

\input{intro}

\input{prelim}
\input{enumerative}
\input{labelcover}
\input{combopt}
\input{conclusion}

\input{acknowledgment}

\printbibliography[heading=bibintoc] 

\appendix
\input{gapvapprox}
\input{gapreduction}

\input{w1-inapprox}
\input{appendix}
\input{gap-eth}

\end{document}

%% file: commands.tex
\newcommand{\poly}{\operatorname{poly}}
\newcommand{\polylog}{\operatorname{polylog}}
\newcommand{\Ostar}{O^\star}
\newcommand{\reals}{\mathbb{R}}
\newcommand{\OPT}{{\sf OPT}\xspace}
\newcommand{\cost}{{\sf COST} \xspace}
\newcommand{\sol}{{\sf SOL} \xspace}
\newcommand{\goal}{{\sf GOAL} \xspace}

\newcommand{\mset}{{\mathcal M}}

\newcommand{\aset}{{\mathcal A}}

\newcommand{\sset}{{\mathcal S}} 
\newcommand{\uset}{{\mathcal U}}
\newcommand{\iset}{{\mathcal I}}
\newcommand{\A}{{\mathbb A}}
\newcommand{\B}{{\mathbb B}}
\newcommand{\N}{{\mathbb N}}
\newcommand{\bG}{\bar{G}}

\newcommand{\MaxCov}{{\sf MaxCov}\xspace}
\newcommand{\MinLab}{{\sf MinLab}\xspace}
\newcommand{\MinRep}{{\sf MinRep}\xspace}
\newcommand{\opt}{{\sf OPT}\xspace}
\newcommand{\Clique}{\mbox{\sf Clique}\xspace}
\newcommand{\DomSet}{\mbox{\sf DomSet}\xspace}
\newcommand{\SetCov}{\mbox{\sf SetCov}\xspace}
\newcommand{\SAT}{\mbox{\sf SAT}\xspace}
\newcommand{\CSP}{\mbox{\sf CSP}\xspace}

\newcommand{\MIS}{\mbox{\sf MIS}\xspace}
\newcommand{\BiClique}{\mbox{\sf Biclique}\xspace}

\newcommand{\IPath}{\mbox{\sf InducedPath}\xspace}
\newcommand{\DkS}{\mbox{\sf DkS}\xspace}
\newcommand{\IM}{\mbox{\sf IM}\xspace}
\newcommand{\Den}{\mbox{\sf Den}\xspace}
\newcommand{\Denk}{\mbox{\sf Den}\textsubscript{$k$}\xspace}

\renewcommand{\P}{\mbox{\sf P}}
\newcommand{\NP}{\mbox{\sf NP}}
\newcommand{\BPP}{\mbox{\sf BPP}}

\newcommand{\ZPP}{\mbox{\sf ZPP}}
\newcommand{\DTIME}{\mbox{\sf DTIME}}
\newcommand{\BPTIME}{\mbox{\sf BPTIME}}

\newcommand{\FPT}{\mbox{\sf FPT}\xspace}
\newcommand{\W}{\mbox{\sf W}\xspace}

\newcommand{\E}{\mathop{\mathbb{E}}}
\newcommand{\rand}{\text{rand}}

\ifdefined\ShowComment

\def\danupon#1{\marginpar{$\leftarrow$\fbox{D}}\footnote{$\Rightarrow$~{\sf #1 --Danupon}}}

\def\parinya#1{\marginpar{$\leftarrow$\fbox{P}}\footnote{$\Rightarrow$~{\sf #1 --Parinya}}}

\else

\def\danupon#1{}
\def\parinya#1{}

\fi

%% file: abstract.tex
\begin{abstract}
We consider questions that arise from the intersection between the
areas of polynomial-time approximation algorithms, subexponential-time
algorithms, and fixed-parameter tractable algorithms.
The questions, which have been asked several times 
(e.g., \cite{Marx08,FellowGMS12,DowneyF13}), are whether there is a
non-trivial {\em FPT-approximation} algorithm for the Maximum Clique
(\Clique) and Minimum Dominating Set (\DomSet) problems parameterized
by the size of the optimal solution. 
In particular, letting $\opt$ be the optimum and $N$ be the size
of the input, is  there an algorithm that runs in
$t(\opt)\poly(N)$ time and outputs a solution of size $f(\opt)$, for
any functions $t$ and $f$ that are independent of $N$ 
(for \Clique, we want $f(\opt)=\omega(1)$)?
%
	
In this paper, we show that both \Clique and \DomSet 
admit no non-trivial FPT-approximation algorithm, i.e.,
there is no $o(\opt)$-FPT-approximation algorithm for \Clique and 
no $f(\opt)$-FPT-approximation algorithm for \DomSet, for any function
$f$ (e.g., this holds even if $f$ is an exponential or the Ackermann
function). In fact, our results imply something even stronger: 
The best way to solve \Clique and \DomSet, even approximately, is to
essentially enumerate all possibilities.
Our results hold under the {\em Gap Exponential Time Hypothesis}
(Gap-ETH)~\cite{Dinur16,ManR16}, which states that no $2^{o(n)}$-time
algorithm can distinguish between a satisfiable 3\SAT formula and 
one which is not even $(1 - \varepsilon)$-satisfiable for some constant $\varepsilon > 0$.

Besides \Clique and \DomSet, we also rule out non-trivial
FPT-approximation for Maximum Biclique, the problem of finding
maximum subgraphs with hereditary properties (e.g., Maximum Induced Planar Subgraph),
and Maximum Induced Matching in bipartite graphs. 
Previously only exact versions of these problems were known to be
$\W[1]$-hard~\cite{Lin15,KhotR00,MoserS09}. 
Additionally, we rule out $k^{o(1)}$-FPT-approximation algorithm for
Densest $k$-Subgraph although this ratio does not yet match the trivial $O(k)$-approximation algorithm.

To the best of our knowledge, prior results only rule out constant
factor approximation for \Clique \cite{HajiaghayiKK13,BonnetE0P15} and 
$\log^{1/4+\epsilon}(\opt)$ approximation for \DomSet for any constant
$\epsilon > 0$ \cite{ChenL16}. 
Our result on \Clique significantly improves on
\cite{HajiaghayiKK13,BonnetE0P15}. 
However, our result on \DomSet is incomparable to \cite{ChenL16} since
their results hold under ETH while our results hold under Gap-ETH, 
which is a stronger assumption. 
\end{abstract}

%% file: intro.tex
\section{Introduction}
\label{sec:intro}



{\em Fixed-parameter approximation algorithm} (in short, FPT-approximation algorithm) is a new concept emerging from a cross-fertilization between
two trends in coping with NP-hard problems: {\em approximation
  algorithms} and {\em fixed-parameter tractable (FPT) algorithms}.
Roughly speaking, an FPT-approximation algorithm is similar to an FPT
algorithm in that its running time can be of the form
$t(\opt)\poly(N)$ time (called the {\em FPT time}), where $t$ is
any function (possibly super exponentially growing), $N$ is the input
size, and $\opt$ is the value of the optimal solution\footnote{There
are many ways to parameterize a problem. In this paper we focus on the
{\em standard parameterization} which parameterizes the optimal
solution.}. 
It is similar to an approximation algorithm in that its output is an 
{\em approximation} of the optimal solution; however, the approximation
factor is analyzed in terms of the optimal solution (\opt) and 
{\em not} the input size ($N$). Thus, an algorithm for a maximization
(respectively, minimization) problem is said to be 
{\em $f(\opt)$-FPT-approximation} for some function $f$ if it outputs
a solution of size at least $\opt/f(\opt)$ (respectively, at most
$\opt\cdot f(\opt)$). \danupon{CHECK: Correct formulation?} 
For a maximization problem, such an algorithm is {\em non-trivial} when $f(\opt)$ is $o(\opt)$, while for a minimization problem, it is non-trivial for any computable function $f$.


The notion of FPT-approximation is useful when we are interested in a
small optimal solution, and in particular its existence connects to a
fundamental question {\em whether there is a non-trivial approximation
algorithm when the optimal solution is small}.
Consider, for example, the {\em Maximum Clique} (\Clique) problem,
where the goal is to find a clique (complete subgraph) with maximum
number of vertices in an $n$-vertex graph $G$. 
%
By outputting any single vertex, we get a trivial polynomial-time 
$n$-approximation algorithm.  
The bound can be improved to $O(\frac{n}{\log n})$ and even to 
$O(\frac{n (\log\log n)^2}{\log^3 n})$ with clever
ideas~\cite{Feige04}. 
Observe, however, that these bounds are quite meaningless when 
$\opt=O(\frac{n (\log\log n)^2}{\log^3 n})$ 
since outputting a single vertex already guarantees such bounds. 
In this case, a bound such as $O(\frac{\opt}{\log\log \opt})$ would be  
more meaningful.
Unfortunately, no approximation ratio of the form $o(\opt)$ is known
even when FPT-time is allowed\footnote{In fact, for maximization
problems, it can be shown that a problem 
admits an $f(\opt)$-FPT-approximation algorithm for some function $f = o(\opt)$
if and only if it admits a polynomial-time algorithm 
with approximation ratio $f'(\opt)$ 
for some function $f' = o(\opt)$ \cite{GroheG07,Marx08} 
(also see \cite{Marx13}). 
So, it does not matter whether the running time is polynomial on
the size of the input or depends on $\opt$.}
(Note that outputting a single vertex gives an $\opt$-approximation
guarantee.)

Similar questions can be asked for a minimization problem. 
Consider for instance, {\em Minimum Dominating Set} (\DomSet): Find the smallest set of vertices $S$ such
that every vertex in an $n$-vertex input graph $G$ has a neighbor in
$S$. 
\DomSet admits an $O(\log n)$-approximation algorithm via
a basic greedy method. 
However, if we want the approximation ratio to depend on $\opt$ and
not $n$, no $f(\opt)$-approximation ratio is known for any 
function $f$ (not even $2^{2^\opt}$).

In fact, the existence of  non-trivial
FPT-approximation algorithms for \Clique and \DomSet has been raised several times
in the literature  (e.g., \cite{Marx08,FellowGMS12,DowneyF13}). 
So far, the progress towards these questions can only rule out
$O(1)$-FPT-approximation algorithms for \Clique.  
This was shown
independently by Hajiaghayi~et~al.~\cite{HajiaghayiKK13} and
Bonnet~et~al.~\cite{BonnetE0P15}, assuming the Exponential Time
Hypothesis (ETH) and that a linear-size PCP exists. 
%
%
%
Alternatively, Khot and Shinkar \cite{KhotS16} proved this under a
rather non-standard assumption that solving quadratic
equations over a finite field under a certain regime of parameters is not in $\FPT$;
unfortunately, this assumption was later shown to be false~\cite{Kayal2014}.
For \DomSet, Chen and Li~\cite{ChenL16} could rule out
$O(1)$-FPT-approximation algorithms assuming $\FPT\neq \W[1]$.
Moreover, they improved the inapproximability ratio to
$\log^{1/4+\epsilon}(\OPT)$ for any constant $\epsilon > 0$
under the exponential time hypothesis (ETH), which asserts that
no subexponential time algorithms can decide whether a given $3$\SAT formula 
is satisfiable. Remark that ETH implies that $\FPT \neq \W[1]$.

%
%
%
%
%

\paragraph{Our Results and Techniques.} 
We show that there is no non-trivial FPT-approximation algorithm for
both \Clique and \DomSet. That is, there is no
$o(\opt)$-FPT-approximation algorithm for \Clique and no
$f(\opt)$-FPT-approximation algorithm for \DomSet, for any function
$f$. Our results hold under the {\em Gap Exponential Time Hypothesis}
(Gap-ETH), which states that distinguishing between a satisfiable $3$\SAT formula 
and one which is not even $(1-\epsilon)$-satisfiable requires 
exponential time for some constant $\epsilon > 0$ (see \Cref{sec:prelim}).
\danupon{TO DO: Make sure that this is stated in prelim}

Gap-ETH, first formalized in \cite{Dinur16,ManR16}, is a
stronger version of the aforementioned ETH. 
It has recently been shown to be useful in proving fine-grained 
hardness of approximation for
problems such as dense CSP with large alphabets~\cite{ManR16} and
Densest-$k$-Subgraph with perfect completeness~\cite{Man17}.

Note that Gap-ETH is implied by ETH if we additionally
assume that a linear-size PCP exists.
So, our result for \Clique significantly improves the results in
\cite{HajiaghayiKK13,BonnetLP16} under the same (in fact, weaker)
assumption. Our result for \DomSet also significantly improves the
results in \cite{ChenL16}, but our assumption is stronger.

In fact, we can show even stronger results: the best way to solve
\Clique and \DomSet, even approximately, is to {\em enumerate all
  possibilities} in the following sense. 
Finding a clique of size
$r$ can be trivially done in $n^{r}\poly(n)$ time by checking
whether any among all possible ${n \choose r}=O(n^r)$ sets of vertices
forms a clique. 
It was known under ETH that this is essentially the
best one can do~\cite{ChenHKX06,ChenHKX06b}. 
We show further that this running time is still needed, even when we
know that a clique of size much larger than $r$ exists in the graph (e.g.,
$\opt \geq 2^{2^r}$), assuming Gap-ETH.  
Similarly, for \DomSet,  we can  always find a dominating set of size $r$ in $n^{r}\poly(n)$ time. Under Gap-ETH, we show that there is no better way even when we just want to find a dominating set of size $q\gg r$. 


We now give an overview of our techniques. 
The main challenge in showing our results is that we want them to hold for the case where the optimal solution is {\em arbitrarily smaller} than the input size.
(This is important to get the FPT-inapproximability results.) To this end, (i) reductions cannot blow up the optimal solution by a function of the input size, and (ii) our reductions must start from problems with a large hardness gap, while having small $\opt$. Fortunately, Property (i) holds for the known reductions we employ.
%
%


The challenge of (ii) is that existing gap amplifying techniques (e.g., the parallel repetition theorem~\cite{Raz98} or the randomized graph product~\cite{BermanS92}), while amplifying the gap to arbitrarily large, cause the input size to be too large that existing $\opt$ reduction techniques (e.g., \cite{ChenHKX06,PatrascuW10}) cannot be applied efficiently (in particular, in subexponential time). 
We circumvent this by a step that amplifies the gap and reduce $\opt$ at the same time. 
In more detail, this step takes a 3SAT formula $\phi$ as an input and produces a ``label cover''\footnote{Our problem is an optimization problem on Label Cover instance, with a slightly different objective from the standard Label Cover. Please refer to Section~\ref{sec:label cover} for more detail.} instance $J$ (roughly, a bipartite graph with constraints on edges) such that: For any $c>0$, (i) If $\phi$ is satisfiable, then $J$ is satisfiable, and (ii) if $\phi$ is at most $0.99$ satisfiable, then less than $c$-fraction of constraints of $J$ can be satisfied. 
Moreover, our reduction allows us to ``compress'' either the the left-hand-side or the right-hand-side vertices to be arbitrarily small. 
This label cover instance is a starting point  for all our problems. 
To derive our result for \Clique, we would need the left-hand-side to be arbitrarily small, while for \DomSet, we would need the small right-hand-side. 

The left-hand-side vertex compression is similar to the randomized graph product~\cite{BermanS92} and, in fact, the reduction itself has been studied before~\cite{Zuck96,Zuck96-unapprox} but in a very different regime of parameters. For a more detailed discussion, please refer to Subsection~\ref{subsec:smallr}.

Once the inapproximability results for label cover problems with small left-hand-side and right-hand-side vertex set are established, we can simply reduce it to \Clique and \DomSet using the standard reductions from~\cite{FGLSS96} and~\cite{Feige98} respectively.

Besides the above results for \Clique and \DomSet, we also show that no non-trivial FPT-approximation algorithm exists for a few other problems, including Maximum Biclique, the problem of finding maximum subgraphs with hereditary properties (e.g., maximum planar induced subgraph) and Maximum Induced Matching in bipartite graphs. Previously only the exact versions of these problems were only known to be $\W[1]$-hard~\cite{Lin15,KhotR00,MoserS09}. Additionally, we rule out $k^{o(1)}$-FPT-approximation algorithm for
Densest $k$-Subgraph, although this ratio does not yet match the trivial $O(k)$-approximation algorithm. Finally, we remark that, while our result for maximum subgraphs with hereditary properties follows from a reduction from \Clique, the FPT inapproximability of other problems are shown not through the label cover problems, but instead from a modification of the hardness of approximation of Densest $k$-Subgraph in~\cite{Man17}. \danupon{If time permits, we may want to expand this paragraph and include more details about previous works.}

\paragraph{Previous Works.}
Our results are based on the method of compressing 
(or reducing the size of) the optimal solution, 
which was first introduced by Chen, Huang, Kanj and  Xia in \cite{ChenHKX04} 
(the journal version appears in \cite{ChenHKX06}).
Assuming ETH, they showed that finding both \Clique and \DomSet cannot be solved in time $n^{o(\OPT)}$,
where $n$ is the number of vertices in an input graph.
Later, P\u{a}tra\cb{s}cu and Williams \cite{PatrascuW10} applied similar techniques to sharpen the running time lower bound of \DomSet to $n^{(1 - \varepsilon)\OPT}$, for any constant $\varepsilon > 0$, assuming
the {\em Strong Exponential Time Hypothesis} (SETH).
The technique of compressing the optimal solution was also used in hardness of approximation 
by Hajiaghayi, Khandekar and Kortsarz in \cite{HajiaghayiKK13}
and by Bonnet, Lampis and Paschos in \cite{BonnetE0P15}.
Our techniques can be seen as introducing gap amplification to the reductions in \cite{ChenHKX06}.
We emphasize that while \cite{ChenHKX06},\cite{PatrascuW10},\cite{HajiaghayiKK13} and \cite{BonnetE0P15} 
(and also the reductions in this paper) are all based on the technique of compressing the optimal solution,
Hajiaghayi~et~al.~\cite{HajiaghayiKK13} compress the optimal solution after reducing \SAT to the 
designated problems, i.e., \Clique and \DomSet. 
\cite{ChenHKX06}, \cite{PatrascuW10}, \cite{BonnetE0P15} and our reductions, on the other hand,
compress the optimal solution of \SAT prior to feeding it to the standard reductions (with small adjustment).
While this difference does not affect the reduction for \Clique,
it has a huge effect on \DomSet.
Specifically, compressing the optimal solution at the post-reduction results in a huge blow-up because the blow-up in the first step (i.e., from \SAT to \DomSet) becomes exponential after compressing the optimal solution.
Our proof for \Clique and the one in \cite{HajiaghayiKK13} bear a similarity in that both apply graph product to amplify 
approximation hardness.
The key different is that we use randomized graph product instead
of the deterministic graph product used in \cite{HajiaghayiKK13}.

Very recently, Chen and Lin \cite{ChenL16} showed that \DomSet admits no
constant approximation algorithm unless $\mathrm{FPT}=\mathrm{W[1]}$.
Their hardness result was derived from the seminal result of Lin
\cite{Lin15}, which shows that the {\em Maximum $k$-Intersection} problem
(a.k.a, {\sf One-side Gap-\BiClique}) has no FPT approximation algorithm. 
Furthermore, they showed that, when assuming ETH, their result can be strengthened to rule out $\log^{1/4+\epsilon}(\OPT)$ FPT-approximation algorithm, for any constant $\epsilon > 0$.
The result of Chen and Lin follows from the W[1]-hardness 
of \BiClique \cite{Lin15} and 
the proof of the ETH-hardness of \Clique \cite{ChenHKX04}.
Note that while Chen and Lin did not discuss the size of the optimal solution in their paper,
the method of compressing the optimal solution was implicitly used there.
This is due to the running-time lower bound of \Clique that they quoted from \cite{ChenHKX04}.

Our method for proving the FPT inapproximability of \DomSet
is similar to that in \cite{PatrascuW10}.
However, the original construction in \cite{PatrascuW10} does not
require a ``partition system''.
This is because P\u{a}tra\cb{s}cu and Williams reduction starts from
\SAT, which can be casted as \DomSet.
In our construction, the reduction starts from an instance of 
the {\em Constraint Satisfaction} problem (\CSP) that is
more general than \SAT (because of the gap-amplification step)
and hence requires the construction of a partition system.
(Note that the partition system has been used in standard hardness reductions for
\DomSet \cite{LundY94,Feige98}.)

We remark that our proof does not imply
FPT-inapproximability for \DomSet under ETH
whereas Chen and Lin were able to prove the inapproximability
result under ETH because their reduction can be applied directly to \SAT 
via the result in \cite{ChenHKX06}.
If ones introduced the Gap-ETH to the previous works,
then the proofs in \cite{ChenHKX06,HajiaghayiKK13,BonnetE0P15}
yield the constant FPT-inapproximability of \Clique,
and the proof in \cite{ChenHKX06} yields
the constant FPT-inapproximability of \DomSet.

The summaries of previous works on \Clique and \DomSet
are presented in \Cref{tab:prev-works}.

\begin{table}[h]
\begin{center}
\begin{tabular}{c|c|c|c}
\multicolumn{4}{c}{\bf Summary of Works on \Clique}\\
\hline
    {\bf \small{Inapprox Factor}}
  & {\bf \small{Running Time Lower Bound}}
  & {\bf \small{Assumption}}
  & {\bf \small{References}}\\
\hline
    any constant
  & $t(OPT) \cdot n^{o(\opt)}$
  & ETH + LPCP
  & \cite{BonnetE0P15}\\
  $\opt^{1-\epsilon}$ 
  & $\exp(\opt^{\rho(\epsilon)})$ 
  & ETH 
  & \cite{ChitnisHK13}\\
  $1/(1-\epsilon)$
  & $\exp(\exp(\opt^{\rho(\epsilon)}))$\footnotemark 
  & ETH 
  & \cite{HajiaghayiKK13}\\
    No $\omega(\opt)$ 
  & $t(OPT) \cdot n^{o(\opt)}$ 
  & Gap-ETH & This paper\\
\hline
\multicolumn{4}{c}{}\\
\multicolumn{4}{c}{\bf Summary of Works on \DomSet}\\
\hline
    {\bf \small{Inapprox Factor}}
  & {\bf \small{Running Time Lower Bound}}
  & {\bf \small{Assumption}}
  & {\bf \small{References}}\\
\hline
  $\opt^{1-\gamma}$ 
  & $\exp(\opt^{1-\rho(\gamma)})$ & ETH & \cite{ChitnisHK13}\\
  $(\log \opt)^{\delta}$ 
  & $\exp(\exp((\log\opt)^{\delta - 1}))$ 
  & ETH + PGC
  & \cite{HajiaghayiKK13}\\
    any constant
  & $t(OPT) \cdot n^{O(1)}$ (i.e. no FPT)
  & $\mathrm{W}[1]\neq\mathrm{FPT}$
  & \cite{ChenL16}\\
    $(\log\opt)^{1/4+\epsilon}$
  & $t(OPT) \cdot n^{o(\sqrt{\opt})}$
  & ETH
  & \cite{ChenL16}+\cite{ChenHKX06}\\
    $f(\opt)$
  & $t(OPT) \cdot n^{o(\opt)}$
  & Gap-ETH
  & This paper\\
\hline
\end{tabular}
\end{center}
\caption{The summaries of works on \Clique and \DomSet.
Here $t$ denotes any computable function $t: \N \to \N$,
	 $\epsilon$ denotes any constant $0<\varepsilon<1$, 
     $\gamma$ denotes some constant $0<\epsilon<1$, 
     $\rho$ denotes some non-decreasing function $\rho:(0,1)\rightarrow(0,1)$, 
     $\delta$ denotes some constant $\delta>1$.
PGC stands for the {\em Projection Game Conjecture} \cite{Moshkovitz15},
and LPCP stands for the {\em Linear-Size PCP Conjecture} \cite{BonnetE0P15}.
}
\label{tab:prev-works}
\end{table}
\footnotetext{Constant FPT-inapproximability of \Clique under ETH is claimed in \cite{HajiaghayiKK13} (arXiv version). However, as we investigated, the Gap-ETH is assumed there.}

\paragraph{Other Related Works.}
All problems considered in this work are also well-studied in terms of hardness of approximation beyond the aforementioned parameterized regimes; indeed many techniques used here are borrowed from or inspired by the non-parameterized settings.

{\bf Maximum Clique.} Maximum Clique is arguably the first natural combinatorial optimization problem studied in the context of hardness of approximation; in a seminal work of Feige, Goldwasser, Lov{\'a}sz, Safra and Szegedy (henceforth FGLSS)~\cite{FGLSS96}, a connection was made between interactive proofs and hardness of approximating \Clique. This connection paves the way for later works on \Clique and other developments in the field of hardness of approximations; indeed, the FGLSS reduction will serve as part of our proof as well. The FGLSS reduction, together with the PCP theorem~\cite{AroraS98,AroraLMSS98} and gap amplification via randomized graph products~\cite{BermanS92}, immediately implies $n^{\varepsilon}$ ratio inapproximability of \Clique for some constant $\varepsilon > 0$ under the assumption that \NP $\subseteq$ \BPP. Following Feige~et~al.'s work, there had been a long line of research on approximability of \Clique~\cite{BellareGLR93,FeigeK00,BellareGS98,BellareS94}, which culminated in H{\aa}stad's work~\cite{Hastad96}. In~\cite{Hastad96}, it was shown that \Clique cannot be approximated to within a factor of $n^{1 - \varepsilon}$ in polynomial time unless \NP $\subseteq$ \ZPP; this was later derandomized by Zuckerman who showed a similar hardness under the assumption \NP $\nsubseteq$ \P~\cite{Zuckerman07}. Since then, better inapproximability ratios are known~\cite{EngebretsenH00,Khot01,KhotP06}, with the best ratio being $n/2^{(\log n)^{3/4 + \varepsilon}}$ for every $\varepsilon > 0$ (assuming \NP $\nsubseteq$ \BPTIME($2^{(\log n)^{O(1)}}$)) due to Khot and Ponnuswami~\cite{KhotP06}. We note here that the best known polynomial time algorithm for \Clique achieves $O\left(\frac{n(\log\log n)^2}{(\log n)^3}\right)$-approximation for the problem~\cite{Feige04}.

{\bf Set Cover.} Minimum Set Cover, which is equivalent to Minimum Dominating Set, is also among the first problems studied in hardness of approximation. Lund and Yannakakis proved that, unless \NP $\subseteq$ \DTIME($2^{(\log n)^{O(1)}}$), \SetCov cannot be efficiently approximated to within $c \log n$ factor of the optimum for some constant $c > 0$~\cite{LundY94}. Not long after, Feige~\cite{Feige98} both improved the approximation ratio and weaken the assumption by showing an $(1 - \varepsilon)\ln n$-ratio inapproximability for every $\varepsilon > 0$ assuming only that \NP $\nsubseteq$ \DTIME($n^{O(\log \log n)}$). Recently, a similar inapproximability has been achieved under the weaker \NP $\nsubseteq$ \P~assumption~\cite{Moshkovitz15,DinurS14}. Since a simple greedy algorithm is known to yield $(\ln n + 1)$-approximation for \SetCov~\cite{Chvatal79}, the aforementioned hardness result is essentially tight. A common feature in all previous works on hardness of \SetCov\cite{LundY94,Feige98,Moshkovitz15} is that the constructions involve composing certain variants of CSPs with partition systems. As touched upon briefly earlier, our construction will also follow this approach; for the exact definition of CSPs and the partition system used in our work, please refer to Subsection~\ref{subsec:domset}.

{\bf Maximum Subgraph with Hereditary Properties.} The complexity of finding and approximating maximum subgraph with hereditary properties have also been studied since the 1980s~\cite{LewisY80,LundY93,feige2005hardness}; specifically, Feige and Kogan showed that, for every non-trivial property $\Pi$ (i.e., $\Pi$ such that infinite many subgraphs satisfy $\Pi$ and infinitely many subgraphs do not satisfy $\Pi$), the problem is hard to approximate to within $n^{1 - \varepsilon}$ factor for every $\varepsilon > 0$ unless \NP $\subseteq$ \ZPP~\cite{feige2005hardness}. We also note that non-trivial approximation algorithms for the problem are known; for instance, when the property fails for some clique or some independent set, a polynomial time $O\left(\frac{n(\log \log n)^2}{(\log n)^2}\right)$-approximation algorithm is known~\cite{Halldorsson00}.

{\bf Maximum Balanced Biclique.} While the Maximum Balanced Biclique problem bears a strong resemblance to the Maximum Clique Problem, inapproximability of the latter cannot be directly translated to that of the former; in fact, despite numerous attempts, not even constant factor NP-hardness of approximation of the Maximum Balanced Biclique problem is known. Fortunately, under stronger assumptions, hardness of approximation for the problem is known: $n^{\varepsilon}$-factor hardness of approximation is known under Feige's random 3\SAT hypothesis~\cite{Feige02} or \NP $\nsubseteq$ $\bigcap_{\varepsilon > 0} $\BPTIME($2^{n^\varepsilon}$)~\cite{Khot06}, and $n^{1 - \varepsilon}$-factor hardness of approximation is known under strengthening of the Unique Games Conjecture~\cite{BhangaleGHKK16,Man17-ICALP}. To the best of our knowledge, no non-trivial approximation algorithm for the problem is known.

{\bf Densest $k$-Subgraph.} The Densest $k$-Subgraph problem has received considerable attention from the approximation algorithm community~\cite{KP93,FPK01,BCCFV10}; the best known polynomial time algorithm due to Bhaskara~et~al.~\cite{BCCFV10} achieves $O(n^{1/4 + \varepsilon})$-approximation for every $\varepsilon > 0$. On the other hand, similar to \BiClique, NP-hardness of approximating Densest $k$-Subgraph, even to some constant ratio, has so far eluded researchers. Nevertheless, in the same works that provide hardness results for \BiClique~\cite{Feige02,Khot06}, \DkS is shown to be hard to approximate to some constant factor under random 3-\SAT hypothesis or \NP $\nsubseteq$ $\bigcap_{\varepsilon > 0} $\BPTIME($2^{n^\varepsilon}$). Furthermore, $2^{\Omega(\log^{2/3} n)}$-factor inapproximability is known under the planted clique hypothesis~\cite{AAMMW11} and, under ETH (resp., Gap-ETH), $n^{1/\poly\log\log n}$ (resp., $n^{o(1)}$) factor inapproximabilities are known~\cite{Man17}. (See also~\cite{BravermanKRW17} in which a constant ratio ETH-hardness of approximating \DkS was shown.) In addition to these hardness results, polynomial ratio integrality gaps for strong LP and SDP relaxations of the problem are also known~\cite{BCVGZ12,M-thesis,ChlamtacMMV17}.

{\bf Maximum Induced Matching on Bipartite Graphs.} The problem was proved to be NP-hard independently by Stockmeyer and Vazirani~\cite{StockmeyerV82} and Cameron~\cite{Cameron89}. The approximability of the problem was first studied by Duckworth~et~al.~\cite{DuckworthMZ05} who showed that the problem is APX-hard, even on bipartite graphs of degree three. Elbassioni~et~al.~\cite{ElbassioniRRS09} then showed that the problem is hard to approximate to within $n^{1/3 - \varepsilon}$ factor for every $\varepsilon > 0$, unless \NP $\subseteq$ \ZPP. Chalermsook~et~al.~\cite{ChalermsookLN13} later improved the ratio to $n^{1 - \varepsilon}$ for every $\varepsilon > 0$.

\paragraph{Organization.} We define basic notations in \Cref{sec:prelim}. In \Cref{sec:inherent}, we define the notion of inherently enumerative, which captures the fact that nothing better than enumerating all possibilities can be done. We show that a problem admits no non-trivial FPT-approximation algorithm by showing that it is inherently enumerative.
In \Cref{sec:label cover}, we define and prove results about our intermediate problems on label cover instances. Finally, in \Cref{sec:hardness-combopt} we derive results for \Clique, \DomSet, and other problems.

%% file: prelim.tex
\section{Preliminaries}
\label{sec:prelim}

We use standard terminology. For any graph $G$, we denote by $V(G)$
and $E(G)$ the vertex and edge sets of $G$, respectively. For each vertex $u \in V(G)$, we denote the set of its neighbors by $N_G(v)$; when the graph $G$ is clear from the context, we sometimes drop it from the notation.
A {\em clique} of $G$ is a complete subgraph of $G$.
Sometime we refer to a clique as a subset $S \subseteq V(G)$ 
such that there is an edge joining every pair of vertices in $S$.
A {\em biclique} of $G$ is a balanced complete bipartite subgraph of $G$
(i.e., the graph $K_{k,k}$).
By $k$-biclique, we mean the graph $K_{k,k}$
(i.e., the number of vertices in each partition is $k$).
An {\em independent set} of $G$ is a subset of vertices $S\subseteq V(G)$
such there is no edge joining any pair of vertices in $S$.
A {\em dominating set} of $G$ is a subset of vertices 
$S\subseteq V(G)$ such that every vertex in $G$ is either in $S$ or
has a neighbor in $S$.
The {\em clique number} (resp., {\em independent number}) of $G$ is
the size of the largest clique (resp., independent set) in $G$.
The {\em biclique number} of $G$ is the largest integer $k$ such that $G$ contains $K_{k, k}$ as a subgraph.
The {\em domination number} of $G$ is defined similarly as 
the size of the smallest dominating set in $G$.
The clique, independent and domination numbers of $G$ are usually
denoted by $\omega(G)$, $\alpha(G)$ and $\gamma(G)$, respectively.
However, in this paper, we will refer to these numbers by
$\Clique(G), \MIS(G), \DomSet(G)$. Additionally, we denote the biclique number of $G$ by $\BiClique(G)$

\subsection{FPT Approximation}

Let us start by formalizing the the notation of optimization problems; here we follow the  notation due to Chen et al.~\cite{ChenGG06}. An \emph{optimization problem} $\Pi$ is defined by three components: (1) for each input instance $I$ of $\Pi$, a set of valid solutions of $I$ denoted by $\sol_\Pi(I)$, (2) for each instance $I$ of $\Pi$ and each $y \in \sol_\Pi(I)$, the cost of $y$ with respect to $I$ denoted by $\cost_\Pi(I, y)$, and (3) the goal of the problem $\goal_\Pi \in \{\min, \max\}$ which specifies whether $\Pi$ is a minimization or maximization problem. Throughout this work, we will assume that $\cost_\Pi(I, y)$ can be computed in time $|I|^{O(1)}$. Finally, we denote by $\opt_{\Pi}(I)$ the optimal value of each instance $I$, i.e. $\opt_\Pi(I) = \goal_\Pi~\cost(I, y)$ where $y$ is taken over $\sol_\Pi(I)$. 

We now continue on to define parameterized approximation algorithms. While our discussion so far has been on optimization problems, we will instead work with ``gap versions'' of these problems. 
Roughly speaking, for a maximization problem $\Pi$, the gap version of $\Pi$ takes in an additional input $k$ and the goal is to decide whether $\OPT_\Pi(I) \geq k$ or $\OPT_\Pi(I) < k / f(k)$. As we will elaborate below, the gap versions are weaker (i.e. easier) than the optimization versions and, hence, our impossibility results for gap versions translate to those of optimization versions as well.


\begin{definition}[FPT gap approximation]\label{def:FPT gap approx}
For any optimization problem $\Pi$ and any computable function $f: \N \rightarrow [1, \infty)$, an algorithm $\A$, which takes as input an instance $I$ of $\Pi$ and a positive integer $k$, is said to be an \emph{$f$-FPT gap approximation algorithm} for $\Pi$ if the following conditions hold on every input $(I, k)$:
\begin{itemize}
\item $\A$ runs in time $t(k) \cdot |I|^{O(1)}$ for some computable function $t: \N \rightarrow \N$.
\item If $\goal_\Pi = \max$, $\A$ must output 1 if $\OPT_\Pi(I) \geq k$ and output 0 if $\OPT_\Pi(I) < k / f(k)$.

If $\goal_\Pi = \min$, $\A$ must output 1 if $\OPT_\Pi(I) \leq k$ and output 0 if $\OPT_\Pi(I) > k \cdot f(k)$.
\end{itemize}
$\Pi$ is said to be \emph{$f$-FPT gap approximable} if there is an $f$-FPT gap approximation algorithm for $\Pi$.
\end{definition}

Next, we formalize the concept of \emph{totally FPT inapproximable}, which encapsulates the non-existence of non-trivial FPT approximations discussed earlier in the introduction.

\begin{definition} \label{def:totalinapprox}
A minimization problem $\Pi$ is said to be \emph{totally FPT inapproximable} if, for every computable function $f: \N \rightarrow [1, \infty)$, $\Pi$ is not $f$-FPT gap approximable. 

A maximization problem $\Pi$ is said to be \emph{totally FPT inapproximable} if, for every computable function $f: \N \rightarrow [1, \infty)$ such that $f(k) = o(k)$ (i.e. $\lim_{k \to \infty} k/f(k) = \infty$), $\Pi$ is not $f$-FPT gap approximable.
\end{definition}

With the exception of Densest $k$-Subgraph, every problem considered in this work will be shown to be totally FPT inapproximable. To this end, we remark that totally FPT inapproximable as defined above through gap problems imply the non-existence of non-trivial FPT approximation algorithm that was discussed in the introduction. These implications are stated more precisely in the two propositions below; their proofs are given in Appendix~\ref{app:gapvapprox}.

\begin{proposition} \label{prop:gapvapprox-min}
Let $\Pi$ be any minimization problem. Then, (1) implies (2) where (1) and (2) are as defined below.
\begin{enumerate}[(1)]
\item $\Pi$ is totally FPT inapproximable.
\item For all computable functions $t: \N \rightarrow \N$ and $f: \N \rightarrow [1, \infty)$, there is no algorithm that, on every instance $I$ of $\Pi$, runs in time $t(\opt_{\Pi}(I)) \cdot |I|^{O(1)}$ and outputs a solution $y \in \sol_\Pi(I)$ such that $\cost_\Pi(I, y) \leq \opt_{\Pi}(I) \cdot f(\opt_{\Pi}(I))$.
\end{enumerate}
\end{proposition}

\begin{proposition} \label{prop:gapvapprox-max}
Let $\Pi$ be any maximization problem. Then, (1) implies (2) where (1) and (2) are as defined below.
\begin{enumerate}[(1)]
\item $\Pi$ is totally FPT inapproximable.
\item For all computable functions $t: \N \rightarrow \N$ and $f: \N \rightarrow [1, \infty)$ such that $f(k) = o(k)$ and $k/f(k)$ is non-decreasing, there is no algorithm that, on every instance $I$ of $\Pi$, runs in time $t(\opt_{\Pi}(I)) \cdot |I|^{O(1)}$ and outputs a solution $y \in \sol_\Pi(I)$ such that $\cost_\Pi(I, y) \geq \opt_{\Pi}(I) / f(\opt_{\Pi}(I))$.
\end{enumerate}
\end{proposition}

\subsection{List of Problems}
We will now list the problems studied in this work. While all the problems here can be defined in terms of optimization problems as defined the previous subsection, we will omit the terms $\sol, \cost$ and $\goal$ whenever they are clear from the context.

\medskip

\noindent{\bf The Maximum Clique Problem (\Clique).}
In $k$-\Clique, we are given a graph $G$ together with an integer $k$, and the goal is to decide 
whether $G$ has a clique of size $k$.
The maximization version of \Clique, called \mbox{\sf Max-Clique},
asks to compute the maximum size of a clique in $G$.
We will abuse \Clique to mean the {\sf Max-Clique} problem,
and we will denote by $\Clique(G)$ the clique number of $G$,
which is the value of the optimal solution to \Clique.

The problem that is (computationally) equivalent to \Clique is 
the {\em maximum independent set} problem (\MIS) which asks to
compute the size of the maximum independent set in $G$. 
The two problems are equivalent since any clique in $G$ is
an independent set in the complement graph $\bar{G}$.

\medskip

\noindent{\bf The Minimum Dominating Set Problem (\DomSet).}
In $k$-\DomSet, we are given a graph $G$ together with an integer $k$, and the goal is to decide
whether $G$ has a dominating set of size $k$. 
The minimization version of $k$-\DomSet is called the \DomSet,
which asks to compute the size of the minimum dominating set in~$G$. 

The problem that is equivalent to \DomSet is 
the {\em minimum set cover} problem (\SetCov):
Given a universe $\mathcal{U}$ of $n$ elements and a collection 
$\mathcal{S}$ of $m$ subsets $S_1,\ldots,S_m\subseteq\mathcal{U}$, 
the goal is to find the minimum number of subsets of $\mathcal{S}$ 
whose union equals $\mathcal{U}$.
It is a standard fact that \DomSet is equivalent to \SetCov. See Appendix~\ref{app:trivial-eq} for more detail.



\medskip

\noindent{\bf Maximum Induced Subgraph with Hereditary Properties:} 
A {\em property} $\Pi$ is simply a subset of all graphs. 
We say that $\Pi$ is a {\em hereditary property} if whenever $G \in \Pi$, all induced subgraphs of $G$ are in $\Pi$. 
The Maximum Induced Subgraph problem with Property $\Pi$ asks for a maximum cardinality set $S \subseteq V(G)$ such that $G[S] \in \Pi$. Here $G[S]$ denotes the subgraph of $G$ induced on $S$.
Notice that both \Clique and \MIS belong to this class of problems. 
For more discussions on problems that belong to this class, see Appendix~\ref{app:trivial-eq}. 

\medskip

\noindent{\bf Maximum Induced Matching on Bipartite Graphs:} 
An induced matching $\mset$ of a graph $G = (V, E)$ is a subset of edges $\{(u_1, v_1), \dots, (u_{|\mset|}, v_{|\mset|})\}$ such that there is no cross edge, i.e., $(u_i, u_j), (v_i, v_j), (u_i, v_j) \notin E$ for all $i \ne j$. The induced matching number $\IM(G)$ of graph $G$ is simply the maximum value of $|\mset|$ among all induced matchings $\mset$'s of $G$. In this work, we will be interested in the problem of approximating $\IM(G)$ in bipartite graphs; this is because, for general graphs, the problem is as hard to approximate as \Clique.
(See Appendix~\ref{app:trivial-eq} for more details.) 

\medskip

\noindent{\bf Maximum Balanced Biclique (\BiClique).}
In $k$-\BiClique, we are given a bipartite graph $G$ together with an integer $k$. The goal is to decide 
whether $G$ contains a complete bipartite subgraph (biclique) with $k$ vertices on each side. In other words, we are asked to decide whether $G$ contains $K_{k,k}$ as a subgraph.
The maximization version of \BiClique, called Maximum Balanced Biclique,
asks to compute the maximum size of a balanced biclique in $G$.

\medskip

\noindent{\bf Densest $k$-Subgraph (\DkS).}
In the Densest $k$-Subgraph problem, we are given an integer $k$ and a graph $G = (V, E)$. The goal is to find a subset $S \subseteq V$ of $k$ vertices that induces maximum number of edges. For convenience, we define density of an induced subgraph $G[S]$ to be $\Den(G[S]) \triangleq \frac{E(G[S])}{\binom{|S|}{2}} \in [0, 1]$ and we define the optimal density of \DkS to be $\Denk(G) = \max_{S \subseteq V, |S| = k} \Den(S)$.

\subsection{Gap Exponential Time Hypothesis}

Our results are based on the Gap Exponential Time Hypothesis (Gap-ETH). Before we state the hypothesis, let us recall the definition of 3-\SAT. In $q$-\SAT, we are given a CNF formula $\phi$ in which each clause consists of at most $q$ literals, and the goal is to decide whether $\phi$ is satisfiable. 

\mbox{\sf Max $q$-SAT} is a maximization version of $q$-\SAT which 
asks to compute the maximum number of clauses in $\phi$ that can 
be simultaneously satisfied.
We will abuse $q$-\SAT to mean \mbox{\sf Max $q$-SAT},
and for a formula $\phi$, we use $\SAT(\phi)$ to denote the maximum number of clauses satisfied by any assignment.

The Gap Exponential Time Hypothesis can now be stated in terms of \SAT as follows.

\begin{conjecture}[(randomized) Gap Exponential-Time Hypothesis (Gap-ETH)
\cite{Dinur16,ManR16}]
\label{conj:gap-ETH}
For some constant $\delta,\epsilon > 0$, no algorithm can,  
given a $3$-\SAT formula $\phi$ on $n$ variables and $m=O(n)$ clauses, 
distinguishes between the following cases correctly with probability $\geq 2/3$ in $O(2^{\delta n})$ time:
\begin{itemize}
	\item $\SAT(\phi) = m$ and  
	\item $\SAT(\phi) < (1-\epsilon) m$.
\end{itemize}
\end{conjecture}

Note that the case where $\epsilon=1/m$ (that is, the algorithm only needs to distinguish between the cases that $\SAT(\phi)=m$ and $\SAT(\phi)<m$)
is known as ETH \cite{IPZ01}. 
Another related conjecture is the strengthened version of ETH is called 
{\em the Strong Exponential-Time Hypothesis} (SETH) \cite{IP01-ETH}:
for any $\epsilon>0$, there is 
an integer $k \geq 3$ such that there is no $2^{(1-\epsilon)n}$-time
algorithm for $k$-SAT.
Gap-ETH of course implies ETH, but, to the best of our knowledge, no formal relationship is known between Gap-ETH and SETH. While Gap-ETH may seem strong due to the gap between the two cases, there are evidences suggesting that it may indeed be true, or, at the very least, refuting it is beyond the reach of our current techniques. We discuss some of these evidences in Appendix~\ref{app:gap-eth}.

While Gap-ETH as stated above rules out not only deterministic but also randomized algorithms, the deterministic version of Gap-ETH suffices for some of our results, including inapproximability of \Clique and \DomSet. The reduction for \DomSet as stated below will already be deterministic, but the reduction for \Clique will be randomized. However, it can be easily derandomized and we sketch the idea behind this in in Subsection~\ref{subsec:derandomization}. Note that, on the other hand, we do not know how to derandomize some of our other results, including those of \BiClique and \DkS.

%% file: enumerative.tex
\section{FPT Inapproximability via Inherently Enumerative Concept}\label{sec:inherent}

Throughout the paper, we will prove FPT inapproximability through the concept of {\em inherently enumerative} problems, which will be formalized shortly.

To motivate the concept, note that all problems $\Pi$ considered in this paper admit an exact algorithm that runs in time\footnote{Recall that $O^{\star}(\cdot)$ hides terms that are polynomial in the input size.} $O^{\star}(|I|^{\opt_{\Pi}(I)})$; For instance, to find a clique of size $k$ in $G$, one can enumerate all ${|V(G)| \choose k} = |V(G)|^{O(k)}$ possibilities\footnote{A faster algorithm runs in time $|V(G)|^{\omega k/3}$ can be done by a reduction to matrix multiplication.}. 
For many W[1]-hard problems (e.g. \Clique), this running time is nearly the best possible assuming ETH: Any algorithm that finds a $k$-clique in time $|V(G)|^{o(k)}$ would break ETH. In the light of such result, it is natural to ask the following question.

\begin{quote}
Assume that $\Clique(G) \geq 2^{2^k}$, can we find a clique of size $k$ in time $|V(G)|^{o(k)}$? 
\end{quote}

In other words, can we exploit a prior knowledge that there is a clique of size much larger than $k$ to help us find a $k$-clique faster? Roughly speaking, we will show later that, assuming Gap-ETH, the answer of this question is also negative, even when $2^{2^k}$ is replaced by any constant independent of $k$. This is encapsulated in the inherently enumerative concept as defined below.

\begin{definition}[Inherently Enumerative] A problem $\Pi$ is said to be {\em inherently enumerative} if there exist constants $\delta, r_0 >0$ such that, for any integers $q \geq r \geq r_0$, no algorithm can decide, on every input instance $I$ of $\Pi$, whether (i) $\opt_{\Pi}(I) < r$ or (ii) $\opt_{\Pi}(I) \geq q$ in time\footnote{$O_{q, r}(\cdot)$ here and in Definition~\ref{def:weakine} hides any multiplicative term that is a function of $q$ and $r$.} $O_{q, r}(|I|^{\delta r})$.  
\end{definition}

While we will show that \Clique and \DomSet are inherently enumerative, we cannot do the same for some other problems, such as \BiClique. Even for the exact version of \BiClique, the best running time lower bound known is only $|V(G)|^{\Omega(\sqrt{k})}$~\cite{Lin15} assuming ETH. 
In order to succinctly categorize such lower bounds, we define a similar but weaker notation of \emph{weakly} inherently enumerative:

\begin{definition}[Weakly Inherently Enumerative] \label{def:weakine}
For any function $\beta = \omega(1)$ (i.e. $\lim_{r \to \infty} \beta(r) = \infty$), a problem $\Pi$ is said to be {\em $\beta$-weakly inherently enumerative} if there exists a constant $r_0 >0$ such that, for any integers $q \geq r \geq r_0$, no algorithm can decide, on every input instance $I$ of $\Pi$, whether (i) $\opt_{\Pi}(I) < r$ or (ii) $\opt_{\Pi}(I) \geq q$ in time $O_{q, r}(|I|^{\beta(r)})$.

$\Pi$ is said to be {\em weakly inherently enumerative} if it is $\beta$-weakly inherently enumerative for some $\beta = \omega(1)$.
\end{definition}

It follows from the definitions that any inherently enumerative problem is also weakly inherently enumerative. As stated earlier, we will prove total FPT inapproximability through inherently enumerative; the proposition below formally establishes a connection between the two. 

\begin{proposition}
If $\Pi$ is weakly inherently enumerative, then $\Pi$ is totally FPT inapproximable.  
\end{proposition}

\begin{proof}
We first consider maximization problems. We will prove the contrapositive of the statement. Assume that a maximization problem $\Pi$ is not totally FPT inapproximable, i.e., $\Pi$ admits an $f$-FPT gap approximation algorithm $\A$ for some computable function $f$ such that $\lim_{k \to \infty} k/f(k) = \infty$. Suppose that the running time of $\A$ on every input $(I, k)$ is $t(k) \cdot |I|^D$ for some constant $D$ and some function $t$. We will show that $\Pi$ is not weakly inherently enumerative.


Let $r_0 > 0$ be any constant and let $\beta: \N \rightarrow \mathbb{R}^+$ be any function such that $\beta = \omega(1)$. Let $r$ be the smallest integer such that $r > r_0$ and $\beta(r) \geq D$ and let $q$ be the smallest integer such that $q/f(q) > r$. Note that $r$ and $q$ exists since $\lim_{r \to \infty} \beta(r) = \infty$ and $\lim_{q \to \infty} q/f(q) = \infty$.

Given any instance $I$ of $\Pi$. From the definition of $f$-FPT gap approximation algorithms (\Cref{def:FPT gap approx}) and from the fact that $q/f(q) > r$, $\A$ on the input $(I, q)$ can  distinguish between $\opt_\Pi(I) \geq q$ and $\opt_\Pi(I) < r$ in $t(q) \cdot |I|^{D} \leq t(q) \cdot |I|^{\beta(r)} = O_{q, r}(|I|^{\beta(r)})$ time. Hence, $\Pi$ is not weakly inherently enumerative, concluding our proof for maximization problems.

%


For any minimization problem $\Pi$, assume again that $\Pi$ is not totally FPT inapproximable, i.e., $\Pi$ admits an $f$-FPT gap approximation algorithm $\A$ for some computable function $f$. Suppose that the running time of $\A$ on every input $(I, k)$ is $t(k) \cdot |I|^D$ for some constant $D$.

Let $r_0 > 0$ be any constant and let $\beta: \N \rightarrow \mathbb{R}^+$ be any function such that $\beta = \omega(1)$. Let $r$ be the smallest integer such that $r > r_0$ and $\beta(r) \geq D$ and let $q = \lceil r \cdot f(r) \rceil$. 

Given any instance $I$ of $\Pi$. From definition of $f$-FPT gap approximation algorithms and from $q \geq r \cdot f(r)$, $\A$ on the input $(I, r)$ can distinguish between $\opt_\Pi(I) \geq q$ and $\opt_\Pi(I) < r$ in $t(r) \cdot |I|^{D} \leq t(r) \cdot |I|^{\beta(r)} = O_{q, r}(|I|^{\beta(r)})$ time. Hence, $\Pi$ is not weakly inherently enumerative.
\end{proof}

An important tool in almost any branch of complexity theory, including parameterized complexity, is a notion of reductions. For the purpose of facilitating proofs of totally FPT inapproximability, we define the following reduction, which we call \emph{FPT gap reductions}.

\begin{definition}[FPT gap reduction]\label{def:FPT gap reduction}
For any functions $f, g = \omega(1)$, a problem $\Pi_0$ is said to be $(f, g)$-\emph{FPT gap reducible} to a problem $\Pi_1$ if there exists an algorithm $\A$ which takes in an instance $I_0$ of $\Pi_0$ and integers $q, r$ and produce an instance $I_1$ of $\Pi_1$ such that the following conditions hold.

\begin{itemize}
\item $\A$ runs in time $t(q, r) \cdot |I_0|^{O(1)}$ for some computable function $t: \N \times \N \to \N$. 
\item For every positive integer $q$, if $\opt_{\Pi_0}(I_0) \geq q$, then $\opt_{\Pi_1}(I_1) \geq f(q)$.
\item For every positive integer $r$, if $\opt_{\Pi_0}(I_0) < g(r)$, then $\opt_{\Pi_1}(I_1) < r$.
\end{itemize}
\end{definition}


It is not hard to see that FPT gap reduction indeed preserves totally FPT inapproximability, as formalized in Proposition~\ref{prop:gapred-totallyinapprox} below. The proof of the proposition can be found in Appendix~\ref{app:gapreduction}.

\begin{proposition} \label{prop:gapred-totallyinapprox}
If a problem $\Pi_0$ is (i) $(f, g)$-FPT gap reducible to $\Pi_1$ for some computable non-decreasing functions $f, g = \omega(1)$, 
and (ii) totally FPT inapproximable, then $\Pi_1$ is also totally FPT inapproximable. 
\end{proposition}

As stated earlier, we mainly work with inherently enumerative concepts instead of working directly with totally FPT inapproximability; indeed, we will never use the above proposition and we alternatively use FPT gap reductions to prove that problems are weakly inherently enumeratives. For this purpose, we will need the following proposition.

\begin{proposition} \label{prop:gapred-enum}
If a problem $\Pi_0$ is (i) $(f, g)$-FPT gap reducible to $\Pi_1$ and (ii) $\beta$-weakly inherently enumerative for some $f, g, \beta = \omega(1)$, then $\Pi_1$ is $\Omega(\beta \circ g)$-weakly inherently enumerative. 
\end{proposition}


\begin{proof}
	We assume that (i) holds, and will show that if the ``then'' part does not hold, then (ii) also does not hold. 
	Recall from \Cref{def:FPT gap reduction} that (i) implies that there exists $C, D > 0$ such that the reduction from $\Pi_0$ (with parameters $q$ and $r$) to $\Pi_1$ takes $O_{q, r}(|I_0|^C)$ time and always output an instance $I_1$ of size at most $O_{q, r}(|I_0|^D)$ on every input instance $I_0$. Now assume that  the ``then'' part does {\em not} hold, in particular  $\Pi_1$ is {\em not} $(\beta \circ g)/D$-weakly inherently enumerative.	We will show the following claim which says that (ii) does not hold (by \Cref{def:weakine}). 
	%
\begin{claim}
For every $r_0 > 0$, there exists $q \geq r \geq r_0$ and an $O_{q, r}(|I_0|^{\beta(r)})$-time algorithm $\B$ that can, on every input instance $I_0$ of $\Pi_0$, distinguish between $\opt_{\Pi_0}(I_0) \geq q$ and $\opt_{\Pi_0}(I_0)<r.$
\end{claim}
	



We now prove the claim.	
Consider any $r_0$.	Since $\beta, g =\omega(1)$, there exists $r'_0$ such that $g(r') \geq r_0$ and $\beta(r')\geq C$, for all $r' \geq r'_0$.
	%
	From the assumption that $\Pi_1$ is not $(\beta \circ g)/D$-weakly inherently enumerative, there exist $q' \geq r' \geq r'_0$ such that there is an $O_{q', r'}(|I_1|^{\beta(g(r'))/D})$-time algorithm $\A$ that can, on every input instance $I_1$ of $\Pi_1$, distinguish between $\opt_{\Pi_1}(I_1) \geq q'$ and $\opt_{\Pi_1}(I_1) < r'$.
	
	Let $r = g(r')$, and let $q$ be the smallest integer such that $f(q) \geq q'$ and $q \geq r$; 
	Note that $q$ exists since $\lim_{q \to \infty}f(q) = \infty$, and that $r\geq r_0$. 
	We use $\A$ and the reduction to build an algorithm $\B$ as follows.
	%
	On input $I_0$, algorithm $\B$ runs the reduction on $I_0$ and the previously defined $q, r$. Let us call the output of the reduction $I_1$. $\B$ then runs $\A$ on input $(I_1, q', r')$ and outputs accordingly; i.e. if $\A$ says that $\opt_{\Pi_1}(I_1) \geq q'$, then $\B$ outputs $\opt_{\Pi_0}(I_0) \geq q$, and, otherwise, if $\A$ says that $\opt_{\Pi_1}(I_1) < r'$, then $\B$ outputs $\opt_{\Pi_0}(I_0) < r$.

	Now we show that $\B$ can distinguish whether $\opt_{\Pi_0}(I_0) \geq q$ or $\opt_{\Pi_1}(I_1)<r$ as desired by the claim: From our choice of $q$, if $\opt_{\Pi_0}(I_0) \geq q$, then $\opt_{\Pi_1}(I_1) \geq f(q) \geq  q'$. Similarly, from our choice of $r = g(r')$, if $\opt_{\Pi_0}(I_0) < r$, then $\opt_{\Pi_1}(I_1) < r'$. Since $\A$ can distinguish between the two cases, $\B$ can distinguish between the two cases as well.
	
	The total running time of $\B$ is  $O_{q, r}(|I_0|^C) + O_{q', r'}(|I_1|^{\beta(g(r'))/D})$ (the first term is for running the reduction). Since  $I_1$ of size at most $O_{q, r}(|I_0|^D)$, $\beta(r)\geq C$, and $q'$ and $r'$ depend only on $q$ and $r$, the running time can be bounded by $O_{q, r}(|I_0|^{\beta(r)})$ as desired. 
\end{proof}

%% file: labelcover.tex
\section{Covering Problems on Label Cover Instances}\label{sec:label cover}

In this section, we give  intermediate results for the lower bounds on the running time of approximating variants of the {\em label cover} problem, which will be the source of our inapproximability results for \Clique and \DomSet.

\subsection{Problems and Results}

\paragraph{Label cover instance:} \danupon{We should emphasize that this has no projection property and we don't assume that the graph is regular.}
A label cover instance $\Gamma$ consists of $(G, \Sigma_U, \Sigma_V, \Pi)$, where 

\begin{itemize}
	\item  $G = (U, V, E)$ is a bipartite graph between vertex sets $U$ and $V$ and an edge set $E$,
	\item $\Sigma_U$ and $\Sigma_V$ are sets of {\em alphabets} to be assigned to vertices in $U$ and $V$, respectively, and
	\item $\Pi=\{\Pi_e\}_{e\in E}$ is a set of {\em constraints} $\Pi_e\subseteq \Sigma_U\times \Sigma_V$. \danupon{Use ``Relations'' instead of ``Constraints''?}
\end{itemize}




We say that $\Pi$ (or $\Gamma$) has the {\em projection property} if for every edge $uv\in E$ (where $u\in U$ and $V\in v$) and every $\alpha\in \Sigma_U$, there is exactly one $\beta\in \Sigma_V$ such that $(\alpha, \beta)\in \Pi_{uv}$.  

%
%

We will define two combinatorial  optimization problems on an instance of the label cover problem.
These two problems are defined on the same instance as the standard label cover problem. We will briefly discuss how our problems differ from the standard one. 

\paragraph{Max-Cover Problem:} 
A {\em labeling} of the graph, is a pair of mappings $\sigma_U: U\rightarrow \Sigma_U$ and $\sigma_V: V\rightarrow \Sigma_V$. We say that a labeling $(\sigma_U, \sigma_V)$ {\em covers} edge $uv$ if $(\sigma_U(u), \sigma_V(v))\in \Pi_{uv}$. We say that a labeling covers a vertex $u$ if it covers every edge incident to $u$. 
%
%
For any label cover instance $\Gamma$, let $\MaxCov(\Gamma)$ denote the maximum number of vertices in $U$ that can be covered by a labeling; i.e. 
\begin{align*}
\MaxCov(\Gamma) &:= \max_{\sigma_U: U\rightarrow \Sigma_U,\ \sigma_V: V\rightarrow \Sigma_V} |\{ u\in U \mid \mbox{$(\sigma_U, \sigma_V)$ covers $u$}\}|.
\end{align*}

The goal of the Max-Cover problem is to compute $\MaxCov(\Gamma)$.
We remark that the standard label cover problem (e.g. \cite{WilliamsonShmoy-Book}) would try to maximize the number of covered {\em edges}, as opposed to our Max-Cover problem, which seeks to maximize the number of covered {\em vertices}.

\paragraph{Min-Label Problem:} 
A {\em multi-labeling} of the graph, is a pair of mappings $\sigma_U: U\rightarrow {\Sigma}_U$ and  $\hat\sigma_V: V\rightarrow 2^{\Sigma_V}$. We say that $(\sigma_U, \hat{\sigma}_V)$  {\em covers} an edge $uv$, if there exists $\beta\in \hat{\sigma}_V(v)$ such that $(\sigma(u), \beta)\in \Pi_{uv}$. 
%
For any label cover instance $\Gamma$, let $\MinLab(\Gamma)$ denote the minimum number of labels needed to assign to vertices in $V$ in order to cover {\em all} vertices in $U$; i.e. 
\begin{align*}
\MinLab(\Gamma) &:= \min_{(\sigma_U, \hat\sigma_V)} \sum_{v\in V} |\hat{\sigma}_V(v)|\,
\end{align*}
where the minimization is over multi-labelings $(\sigma_U, \hat{\sigma}_V)$  that covers every edge in $G$.

We emphasize that we can assign multiple labels to nodes in $V$ while each node in $U$ must be assigned a unique label. Note that \MinLab is different from the problem known in the literature as \MinRep (e.g. \cite{CharikarHK11}); in particular, in \MinRep we can assign multiple labels to all nodes.


\paragraph{Results.}
First, note that checking whether $\MaxCov(\Gamma)< r$ or not, for any $r\geq 1$, can be done by the following algorithms. 

\begin{enumerate}
	

	\item It can be done\footnote{Recall that we use $\Ostar(\cdot)$ to hide factors polynomial in the input size.} in $\Ostar({|U| \choose r} (|\Sigma_U|)^r) = \Ostar((|U|\cdot|\Sigma_U|)^r)$ time: First, enumerate all ${|U| \choose r}$ possible subsets $U'$ of $U$ and all $|\Sigma_U|^{|U'|}$ possible labelings on vertices in $U'$. Once we fix the labeling on $U'$, we only need polynomial time to check whether we can label other vertices so that all vertices in $U'$ are covered. 
	
	\item It can be done in $\Ostar(|\Sigma_V|^{|V|})$ time: Enumerate all $\Ostar(|\Sigma_V|^{|V|})$ possible labelings $\sigma_V$ on $V$.  After $\sigma_V$ is fixed, we can find labeling $\sigma_U$ on $U$ that maximizes the number of vertices covered in $U$ in polynomial time. 
\end{enumerate}

ETH can be restated as that these algorithms are the best possible when $|U|=\Theta(|V|)$, $|\Sigma_U|, |\Sigma_V|=O(1)$ and $\Pi$ has the projection property. Gap-ETH asserts further that this is the case even to distinguish between $\MaxCov(\Gamma)=|U|$ and $\MaxCov(\Gamma)\leq (1-\varepsilon)|U|$.



\begin{theorem}\label{thm:Gap-ETH restated}
	Gap-ETH (\Cref{conj:gap-ETH}) is equivalent to the following statement. There exist constants $\varepsilon,\delta > 0$ such that no algorithm can take a label cover instance $\Gamma$ and can distinguish between the following cases in $O(2^{\delta |U|})$ time: 
	\begin{itemize}
		\item $\MaxCov(\Gamma)=|U|$, and 
		\item $\MaxCov(\Gamma)<(1-\varepsilon)|U|$.
	\end{itemize} 
  This holds even when $|\Sigma_U|, |\Sigma_V|=O(1)$, $|U|=\Theta(|V|)$ and $\Pi$ has the projection property.\danupon{Should emphasize that $|U|$ can be arbitrarily large.}
\end{theorem}	

The proof of \Cref{thm:Gap-ETH restated} is standard.
To avoid distracting the readers, we provide the sketch of the proof
in \Cref{sec:restate-gap-eth}.

We will show that \Cref{thm:Gap-ETH restated} can be extended to several cases, which will be useful later. 
%
First, consider when the first ($\Ostar((|U|\cdot|\Sigma_U|)^{r})$-time) algorithm is faster than the second.  We show that in this case this algorithm is essentially the best even for $r=O(1)$, and this holds even when we know that $\MaxCov(\Gamma)=|U|$. 

%
%

For convenience, in the statements of \Cref{thm:small r,thm:small V,thm:MinLab} below, we will use the notation $|\Gamma|$ to denote the size of the label cover instance; in particular, $|\Gamma| = |\Sigma_U| |\Sigma_V| |U| |V|$. Furthermore, recall that the notation $O_{k, r}(\cdot)$ denotes any multiplicative factor that depends only on $k$ and $r$.

\begin{theorem}[\MaxCov with Small $|U|$]
\label{thm:small r}
Assuming Gap-ETH, there exist constants $\delta, \rho > 0$ such that, for any positive integers $k \geq r \geq \rho$, no algorithm can take a label cover instance $\Gamma$ with $|U| = k$ and distinguish between the following cases in $O_{k, r}(|\Gamma|^{\delta r})$ time: 
	\begin{itemize}
		\item 	$\MaxCov(\Gamma) = k$ and  
		\item $\MaxCov(\Gamma)< r$.\danupon{Note ``$<$''.}
	\end{itemize}
This holds even when  $|\Sigma_V|=O(1)$ and $\Pi$ has the projection property.
\end{theorem}

\danupon{NEXT TIME: Say that if we assume ETH then we get similar thing except that $r\geq ...$}

We emphasize that it is important for applications in later sections that $r=O(1)$. In fact, the main challenge in proving the theorem above is to prove it true for $r$ that is arbitrarily small compared to $|U|$. 

Secondly, consider when the second ($\Ostar(|\Sigma_V|^{|V|})$-time) algorithm is faster; in particular when $|V|\ll |U|$. In this case, we cannot make the soundness (i.e. parameter $r$ in  \Cref{thm:small r}) to be arbitrarily small. (Roughly speaking, the first algorithm can become faster otherwise.) 
Instead, we will  show that the second algorithm is essentially the best possible for soundness as small as $\gamma |U|$, for any constant $\gamma >0$. More importantly, this holds for $|V|=O(1)$ (thus independent from the input size). This is the key property of this theorem that we need later. 

\begin{theorem}[\MaxCov with Small $|V|$]
\label{thm:small V}
Assuming Gap-ETH, there exist constants $\delta, \rho > 0$ such that, for any positive integer $q \geq \rho$ and any $1 \geq \gamma > 0$, no algorithm can take a label cover instance $\Gamma$ with $|V|=q$ and distinguish between the following cases in $O_{q, \gamma}(|\Gamma|^{\delta q})$ time: 
	\begin{itemize}
		\item $\MaxCov(\Gamma)= |U|$ and  
		\item $\MaxCov(\Gamma) < \gamma |U|$.
	\end{itemize}
This holds even when $|\Sigma_U|\leq (1/\gamma)^{O(1)}$.
\end{theorem}

We remark that the above label cover instance does not have the projection property. 

In our final result, we turn to computing $\MinLab(\Gamma)$. Since $\MaxCov(\Gamma)=|U|$ if and only if $\MinLab(\Gamma)=|V|$, a statement similar to \Cref{thm:Gap-ETH restated} intuitively holds for distinguishing between  $\MinLab(\Gamma)\leq |V|$ and $\MinLab(\Gamma) > (1+\varepsilon) |V|$; i.e. we need $\Ostar(|\Sigma_V|^{|V|})$ time. 
In the following theorem, we show that this gap can be substantially amplified, while maintaining the property that $|V|=O(1)$ (thus independent from the input size).


\begin{theorem}[\MinLab Hardness]
\label{thm:MinLab}
Assuming Gap-ETH, there exist constants $\delta, \rho > 0$ such that, for any positive integers $r \geq q \geq \rho$,  no algorithm can take a label cover instance $\Gamma$ with $|V|=q$, and distinguish between the following cases in $O_{q, r}(|\Gamma|^{\delta q})$ time: 
\begin{itemize}
	\item 	$\MinLab(\Gamma)= q$ and  
	\item $\MinLab(\Gamma) >r$.
\end{itemize}
This holds even when $|\Sigma_U|= (r/q)^{O(q)}$. 
\end{theorem}

The rest of this section is devoted to proving \Cref{thm:small r,thm:small V,thm:MinLab}.


\subsection{Proof of \Cref{thm:small r}} \label{subsec:smallr}

The proof proceeds by {\bf compressing the left vertex set $U$} of a
label cover instance from \Cref{thm:Gap-ETH restated}. More specifically, each new left vertex will be a subset of left vertices in the original instance. In the construction below, these subsets will just be random subsets of the original vertex set of a certain size; however, the only property of random subsets we will need is that they form a \emph{disperser}. To clarify our proof, let us start by stating the definition of dispersers here. Note that, even though dispersers are often described in graph or distribution terminologies in literatures (e.g.~\cite{Vadhan-book}), it is more convenient for us to describe it in terms of subsets.

\begin{definition}
For any positive integers $m, k, \ell, r \in \N$ and any constant $\varepsilon \in (0, 1)$, an \emph{$(m, k, \ell, r, \varepsilon)$-disperser} is a collection $\cI$ of $k$ subsets $I_1, \dots, I_k \subseteq [m]$ each of size $\ell$ such that the union of any $r$ different subsets from the collection has size at least $(1 - \varepsilon)m$. In other words, for any $1 \leq i_1 < \cdots < i_r \leq k$, we have $|I_{i_1} \cup \cdots \cup I_{i_r}| \geq (1 - \varepsilon)m$.
\end{definition}

The idea of using dispersers to amplify gap in hardness of approximation bears a strong resemblance to the classical randomized graph product technique~\cite{BermanS92}. Indeed, similar approaches have been used before, both implicitly (e.g.~\cite{BellareGS98}) and explicitly (e.g.~\cite{Zuck96,Zuck96-unapprox,Zuckerman07}). In fact, even the reduction we use below has been studied before by Zuckerman~\cite{Zuck96,Zuck96-unapprox}!

What differentiates our proof from previous works is the setting of parameters. Since the reduction size (specifically, the left alphabet size $|\Sigma_U|$) blows up exponentially in $\ell$ and previous results aim to prove NP-hardness of approximating \Clique, $\ell$ are chosen to be small (i.e. $O(\log m)$). On the other hand, we will choose our $\ell$ to be $\Theta_{\varepsilon}(m/r)$ since we would like to only prove a running time lower bound of the form $|\Sigma_U|^{\Omega(r)}$. Interestingly, dispersers for our regime of parameters are easier to construct deterministically and we will sketch the construction in Subsection~\ref{subsec:derandomization}. Note that this construction immediately implies derandomization of our reduction.

The exact dependency of parameters can be found in the claim below, which also states that random subsets will be a disperser for such choice of parameters with high probability. Here and throughout the proof, $k$ and $r$ should be thought of as constants where $k \gg r$; these are the same $k, r$ as the ones in the statement of \Cref{thm:small r}.

\begin{claim} \label{claim:random-disperser}
For any positive integers $m, k, r \in \N$ and any constant $\varepsilon \in (0, 1)$, let $\ell = \max\{m, \lceil 3m/(\varepsilon r)\rceil\}$ and let $I_1, \dots, I_k$ be $\ell$-element subsets of $[m]$ drawn uniformly independently at random. If $\ln k \leq m/r$, then $\cI = \{I_1, \dots, I_k\}$ is an $(m, k, \ell, r, \varepsilon)$-disperser with probability at least $1 - e^{-m}$.
\end{claim}

\begin{proof}
When $\ell = m$, the statement is obviously true; thus, we assume w.l.o.g. that $\ell = \lceil 3m/(\varepsilon r)\rceil$. Consider any indices $i_1, \dots, i_r$ such that $1 \leq i_1 < \cdots < i_r \leq k$. We will first compute the probability that $|I_{i_1} \cup \cdots \cup I_{i_r}| < (1 - \varepsilon)m$ and then take the union bound over all such $(i_1, \dots, i_r)$'s.

Observe that $|I_{i_1} \cup \cdots \cup I_{i_r}| < (1 - \varepsilon)m$ if and only if there exists a set $S \subseteq [m]$ of size less than $(1 - \varepsilon)m$ such that $I_{i_1}, \dots, I_{i_r} \subseteq S$. For a fixed set $S \subseteq [m]$ of size less than $(1 - \varepsilon)m$, since $I_{i_1}, \dots, I_{i_r}$ are independently drawn random $\ell$-element subsets of $[m]$, we have
\begin{align*} \label{eq:}
\Pr[I_{i_1}, \dots, I_{i_r} \subseteq S]
= \prod_{j \in [r]} \Pr[I_j \subseteq S]
= \left(\frac{\binom{|S|}{\ell}}{\binom{m}{\ell}}\right)^r
\leq \left(\frac{|S|}{m}\right)^{\ell r}
< (1 - \varepsilon)^{\ell r}
\leq e^{-\varepsilon \ell r}
< e^{-3m}.
\end{align*}
Taking the union bound over all such $S$'s, we have
\begin{align*}
\Pr[|I_{i_1} \cup \dots \cup I_{i_r}| < (1 - \varepsilon) m]
< \sum_{S \subseteq [m], |S| < (1 - \varepsilon) m} e^{-3m}
< 2^m \cdot e^{-3m}
< e^{-2m}.
\end{align*}
Finally, taking the union bound over all $(i_1, \dots, i_r)$'s gives us the desired probabilistic bound:
\begin{align*}
\Pr[\cI \text{ is not an } (m, k, \ell, r, \varepsilon)\text{-disperser}] \leq \sum_{1 \leq i_1 < \cdot < i_r \leq k} e^{-2m} \leq k^r \cdot e^{-2m} < e^{-m},
\end{align*}
where the last inequality comes from our assumption that $\ln k \leq m / r$.
\end{proof}

With the definition of dispersers and the above claim ready, we move on to prove \Cref{thm:small r}.

\begin{proof}[Proof of \Cref{thm:small r}]
First, we take a label cover instance 
$\wtGamma=(\wtG=(\wtU,\wtV,\wtE),\Sigma_{\wtU},\Sigma_{\wtV},\wtPi)$
as in \Cref{thm:Gap-ETH restated}.
We may assume that $|\Sigma_{\wtU}|,|\Sigma_{\wtV}|=O(1)$,
and $|\wtU|=\Theta(|\wtV|)$. Moreover, let $m = |\wtU|$ and $n = |\wtV|$; for convenience, we rename the vertices in $\wtU$ and $\wtV$ so that $\wtU = [m]$ and $\wtV = [n]$.
Note that it might be useful for the readers to think
of $\wtGamma$ as a $3$-\SAT instance where $\wtU$ is the 
set of clauses and $\wtV$ is the set of variables.

We recall the parameter $\varepsilon$ from 
\Cref{thm:Gap-ETH restated} and the parameters $k, r$
from the statement of \Cref{thm:small r}.
We introduce a new parameter $\ell = 3 m / (\varepsilon r)$ and assume w.l.o.g. that $\ell$ is an integer.

The new label cover (\MaxCov) instance $\Gamma=(G=(U,V,E),\Sigma_U,\Sigma_V,\Pi)$ is defined as follows.
\begin{itemize} \setlength\itemsep{0em}
\item The right vertices and right alphabet set remain unchanged, i.e., $V = \wtV$ and $\Sigma_V = \Sigma_{\wtV}$.
\item There will be $k$ vertices in $U$ where each vertex is a random set of $\ell$ vertices of $\wtU$. More specifically, we define $U = \{I_1,\ldots, I_k\}$ where each $I_i$ is a random $\ell$-element subsets of $[m]$ drawn independently of each other.
\item The  left alphabet set $\Sigma_U$ is $\Sigma_{\wtU}^\ell$. For each $I \in U$, we view each label $\alpha \in \Sigma_U$ as a tuple $(\alpha_u)_{u \in I} \in (\Sigma_{\wtU})^I$; this is a partial assignment to all vertices $u \in I$ in the original instance $\wtGamma$.
\item We create an edge between $I \in U$ and $v \in V$ in $E$ if and only if there exists $u \in I$ such that $uv \in \wtE$. More formally, $E = \{Iv: I \cap N_{\wtG}(v) \ne \emptyset\}$.
\item Finally, we define the constraint $\Pi_{Iv}$ for each $Iv \in E$. As stated above, we view each $\alpha \in \Sigma_U$ as a partial assignment $(\alpha_u)_{u \in I}$ for $I \subseteq \wtU$. The constraint $\Pi_{Iv}$ then contains all $(\alpha, \beta)$ such that $(\alpha_u, \beta)$ satisfies the constraint $\wtPi_{uv}$ for every $u \in I$ that has an edge to $v$ in $\wtGamma$. More precisely, $\Pi_{Iv} = \{(\alpha, \beta) = ((\alpha_u)_{u \in I}, \beta): \forall u \in I \cap N_{\wtG}(v), (\alpha_u, \beta) \in \wtPi_{uv}\}$.
\end{itemize}

Readers who prefer the $3$-\SAT/CSP viewpoint of label cover may think of each $I_i$ as a collection of clauses in the $3$-\SAT instance that are joined by an operator {\bf AND}, i.e., the assignment must satisfy all the clauses in $I_i$ simultaneously in order to satisfy $I_i$.

We remark that, if $\wtPi$ has the projection property, $\Pi$ also has projection property.

\medskip
\noindent{\bf Completeness.} 
Suppose there is a labeling $(\sigma_{\wtU},\sigma_{\wtV})$ of
$\wtGamma$ that covers all $|\wtU|$ left-vertices.
We take $\sigma_V=\sigma_{\wtV}$ and construct $\sigma_U$ by setting $\sigma_{U}(I)=(\sigma_{\wtU}(u))_{u \in I}$ for each $I \in U$.
Since $(\sigma_{\wtU},\sigma_{\wtV})$ covers all the vertices of $\wtU$,
$(\sigma_U,\sigma_V)$ also covers all the vertices of $U$. Therefore, $\MaxCov(\Gamma)=|U|$.

\medskip
\noindent{\bf Soundness.}
To analyze the soundness of the reduction, first recall Claim~\ref{claim:random-disperser}: $\{I_1, \dots, I_k\}$ is an $(m, k, \ell, r, \varepsilon)$-disperser with high probability. Conditioned on this event happening, we will prove the soundness property, i.e., that if $\MaxCov(\wtGamma) < (1 - \varepsilon)|\wtU|$, then $\MaxCov(\Gamma) < r$.

We will prove this by contrapositive; assume that there is a labeling $(\sigma_U, \sigma_V)$ that covers at least $r$ left vertices $I_{i_1}, \cdots, I_{i_r} \in U$. We construct a labeling $(\sigma_{\wtU}, \sigma_{\wtV})$ as follows. First, $\sigma_{\wtV}$ is simply set to $\sigma_V$. Moreover, for each $u \in I_{i_1} \cup \cdots \cup I_{i_r}$, let $\sigma_{\wtU}(u) = (\sigma_U(I_{i_j}))_u$ where $j \in [r]$ is an index such that $u \in I_{i_j}$; if there are multiple such $j$'s, just pick an arbitrary one. Finally, for $u \in U \setminus (I_{i_1} \cup \cdots \cup I_{i_r})$, we set $\sigma_{\wtU}(u)$ arbitrarily.

We claim that, every $u \in I_{i_1} \cup \cdots \cup I_{i_r}$ is covered by $(\sigma_{\wtU}, \sigma_{\wtV})$ in the original instance $\wtGamma$. To see that this is the case, recall that $\sigma_{\wtU}(u) = (\sigma_U(I_{i_j}))_u$ for some $j \in [r]$ such that $u \in I_{i_j}$. For every $v \in V$, if $uv \in E$, then, from how the constraint $\Pi_{I_{i_j}v}$ is defined, we have $(\sigma_{\wtU}(u), \sigma_{\wtV}(v)) = (\sigma_U(I_{i_j})_u, \sigma_V(v)) \in \wtPi_{uv}$. In other words, $u$ is indeed covered by $(\sigma_{\wtU}, \sigma_{\wtV})$.

Hence, $(\sigma_{\wtU}, \sigma_{\wtV})$ covers at least $|I_{i_1} \cup \cdots \cup I_{i_r}| \geq (1 - \varepsilon)m$, where the inequality comes from the definition of dispersers. As a result, $\MaxCov(\wtGamma) \geq (1 - \varepsilon)|\wtU|$, completing the soundness proof.

\medskip

\noindent{\bf Running Time Lower Bound.}
Our construction gives a \MaxCov instance $\Gamma$ with $|U|=k$ and
$|\Sigma_U|=|\Sigma_{\wtU}|^\ell=2^{\Theta(m/(\varepsilon r))}$, whereas $|V|$ and $|\Sigma_V|$ remain $n$ and $O(1)$ respectively. Assume that Gap-ETH holds and let $\delta_0$ be the constant in the running time lower bound in \Cref{thm:Gap-ETH restated}. Let $\delta$ be any constant such that $0 < \delta < \delta_0 \varepsilon / c$ where $c$ is the constant such that $|\Sigma_U| \leq 2^{cm/(\varepsilon r)}$.

Suppose for the sake of contradiction that, for some $k \geq r \geq \rho$, there is an algorithm that distinguishes whether $\MaxCov(\Gamma) = k$ or $\MaxCov(\Gamma) < r$ in $O_{k, r}(|\Gamma|^{\delta r})$ time. Observe that, in our reduction, $|U|, |V|, |\Sigma_V| = |\Sigma_U|^{o(1)}$. Hence, the running time of the algorithm on input $\Gamma$ is at most $O_{k, r}(|\Sigma_U|^{\delta r(1 + o(1))}) \leq O_{k, r}(|\Sigma_U|^{\delta_0 \varepsilon r / c}) \leq O(2^{\delta_0 m})$ where the first inequality comes from our choice of $\delta$ and the second comes from $|\Sigma_U| \leq 2^{cm/(\varepsilon r)}$. Thanks to the completeness and soundness of the reduction, this algorithm can also distinguish whether $\MaxCov(\wtGamma) = |\wtU|$ or $\MaxCov(\wtGamma) < (1 - \varepsilon)|\wtU|$ in time $O(2^{\delta_0 m})$. From \Cref{thm:Gap-ETH restated}, this is indeed a contradiction.
\end{proof}

\subsubsection{Derandomization} \label{subsec:derandomization}

While the reduction in the proof of \Cref{thm:small r} is a randomized reduction, it can be derandomized quite easily. We sketch the ideas behind the derandomization below.

Notice that the only property we need from the random $\ell$-element subsets $I_1, \dots, I_k$ is that it forms an $(m, k, \ell, r, \varepsilon)$-disperser. Hence, to derandomize the reduction, it suffices to deterministically construct such a disperser in $2^{o(n)}$ time.

To do so, let us first note that Lemma~\ref{claim:random-disperser} implies that an $(m', k, \ell', r, \varepsilon)$-disperser exists where $m' = r \ln k$ and $\ell' = 3m'/(\varepsilon r)$. For convenience, we assume w.l.o.g. that $m', \ell'$ are integers and that $m'$ divides $m$. Since $m'$ is now small, we can find such a disperser by just enumerating over every possible collection of $k$ subsets of $[m']$ each of size $\ell'$ and checking whether it has the desired property; this takes only $(2^{m'})^k(k)^r\poly(m') = 2^{O(rk \log k)}$ time, which is acceptable for us since $r$ and $k$ are both constants. Let the $(m', k, \ell', r, \varepsilon)$-disperser that we find be $\{I'_1, \dots, I'_k\}$. Finally, to get from here to the intended $(m, k, \ell, r, \varepsilon)$-disperser, we only need to view $[m]$ as $[m/m'] \times [m']$ and let $I_1 = [m/m'] \times I'_1, \dots, I_k = [m/m'] \times I'_k$. It is not hard to check that $\{I_1, \dots, I_k\}$ is indeed an $(m, k, \ell, r, \varepsilon)$-disperser, which concludes our sketch.


\subsection{Proof of \Cref{thm:small V}}

The proof proceeds by {\bf compressing the right vertex set $V$} of a
label cover instance from \Cref{thm:Gap-ETH restated} plus amplifying the
hardness gap. The gap amplification step is similar to that in the proof of
\Cref{thm:small r} except that, since here $\MaxCov(\Gamma)$ is not required to be constant in the soundness case, we can simply take all subsets of appropriate sizes instead of random subsets as in the previous proof; this also means that our reduction is deterministic and requires no derandomization.

\begin{proof}[Proof of \Cref{thm:small V}]
First, we take a label cover instance 
$\wtGamma = (\wtG=(\wtU,\wtV,\wtE),\Sigma_{\wtU},\Sigma_{\wtV},\wtPi)$
as in \Cref{thm:Gap-ETH restated}.
We may assume that $|\Sigma_{\wtU}|,|\Sigma_{\wtV}|=O(1)$,
and $|\wtU|=\Theta(|\wtV|)$. For convenience, we assume w.l.o.g. that $\wtU = [m]$ and $\wtV = [n]$.
Again, it might be useful for the readers to think
of $\wtGamma$ as a $3$-\SAT instance where $\wtU$ are the 
set of clauses and $\wtV$ are the set of variables.

Recall the parameter $\varepsilon$ from 
\Cref{thm:Gap-ETH restated} and the parameters 
$q, \gamma$ from \Cref{thm:small V}.
Let $\ell=\ln(1/\gamma)/\varepsilon$. We assume w.l.o.g. that $\ell$ is an integer and that $n$ is divisible by $q$.
The new label cover (\MaxCov) instance $\Gamma=(G=(U,V,E),\Sigma_U,\Sigma_V,\Pi)$ is defined as follows.
\begin{itemize} \setlength\itemsep{0em}
\item First, we partition $\wtV = [n]$ into $q$ parts $J_1,\ldots,J_q$, 
each of size $n/q$. We then let $V = \{J_1, \dots, J_q\}$. In other words, we merge $n/q$ vertices of $\wtV$ into a single vertex in $V$.
\item Let $U$ be $\binom{[m]}{\ell}$, the collection of all $\ell$-element subsets of $[m] = \wtU$.
\item The left alphabet set $\Sigma_U$ is $\Sigma_{\wtU}^\ell$. For each $I \in U$, we view each label $\alpha \in \Sigma_U$ as a tuple $(\alpha_u)_{u \in I} \in (\Sigma_{\wtU})^I$; this is a partial assignment to all vertices $u \in I$ in the original instance $\wtGamma$.

\item Our graph $G$ is simply a complete bipartite graph, i.e., for every $I \in U$ and $J \in V$, $IJ \in E(G)$.

\item The label set of $V$ is $\Sigma_V=\Sigma_{\wtV}^{n/q}$,
and the label set of $U$ is $\Sigma_U=\Sigma_{\wtU}^\ell$. For each $I \in U$, we view each label $\alpha \in \Sigma_U$ as a tuple $(\alpha_u)_{u \in I} \in (\Sigma_{\wtU})^I$; this is simply a partial assignment to all vertices $u \in I$ in the original instance $\wtGamma$. Similarly, for each $J \in V$, we view each label $\beta \in \Sigma_V$ as $(\beta_v)_{v \in J} \in (\Sigma_{\wtV})^J$.

\item Finally, we define $\Pi_{IJ}$ for each $IJ \in E$. The constraint $\Pi_{IJ}$ contains all $(\alpha, \beta)$ such that $(\alpha_u, \beta_v)$ satisfies the constraint $\wtPi_{uv}$ for every $u \in I, v \in J$ such that $uv \in \wtE$. More precisely, $\Pi_{IJ} = \{(\alpha, \beta) = ((\alpha_u)_{u \in I}, (\beta_v)_{v \in J}): \forall u \in I, v \in J \text{ such that } uv \in \wtE, (\alpha_u, \beta_v) \in \wtPi_{uv}\}$.
\end{itemize}

We remark that $\Pi$ may not have the projection property even when $\wtPi$ has the property.

\medskip
\noindent{\bf Completeness.} 
Suppose that there is a labeling $(\sigma_{\wtU},\sigma_{\wtV})$ of
$\wtGamma$ that covers all $|\wtU|$ left-vertices.
We construct $(\sigma_U, \sigma_V)$ by setting 
$\sigma_U(I) = (\sigma_{\wtU}(u))_{u \in I}$ for each $I \in U$
and $\sigma_V(J) = (\sigma_{\wtV}(v))_{v \in J}$ for each $J \in V$.  
It is easy to see that $(\sigma_U,\sigma_V)$ covers all the vertices of $U$. Therefore, $\MaxCov(\Gamma)=|U|$.

\medskip
\noindent{\bf Soundness.}
Suppose that $\MaxCov(\wtGamma) < (1-\varepsilon)|\wtU|$. Consider any labeling $(\sigma_U,\sigma_V)$ of $\Gamma$; we will show that $(\sigma_U, \sigma_V)$ covers less than $\gamma|U|$ left-vertices.

Let $I_1, \dots, I_t \in U$ be the vertices covered by $(\sigma_U, \sigma_V)$. Analogous to the proof of \Cref{thm:small r}, we define a labeling $(\sigma_{\wtU}, \sigma_{\wtV})$ as follows. First, $\sigma_{\wtV}$ is naturally defined from $\sigma_V$ by $\sigma_{\wtV} = \sigma_V(J)_v$ where $J$ is the partition that contains $v$. Moreover, for each $u \in I_{i_1} \cup \cdots \cup I_{i_r}$, let $\sigma_{\wtU}(u) = (\sigma_U(I_{i_j}))_u$ where $j \in [r]$ is an index such that $u \in I_{i_j}$; for $u \in U \setminus (I_{i_1} \cup \cdots \cup I_{i_r})$, we set $\sigma_{\wtU}(u)$ arbitrarily.

Similar to the proof of \Cref{thm:small r}, it is not hard to see that every vertex in $I_1 \cup \cdots \cup I_t$ is covered by $(\sigma_{\wtU}, \sigma_{\wtV})$ in $\wtGamma$. Since $\MaxCov(\wtGamma) < (1-\varepsilon)|\wtU|$, we can conclude that $|I_1 \cup \cdots \cup I_t| < (1 - \varepsilon)|\wtU|$. Since each $I_i$ is simply an $\ell$-size subset of $I_1 \cup \cdots \cup I_t$, we can conclude that
\begin{align*}
t < \binom{(1 - \varepsilon)|\wtU|}{\ell} \leq (1 - \varepsilon)^{\ell}\binom{|\wtU|}{\ell} = (1 - \varepsilon)^{\ell}|U| \leq e^{-\varepsilon \ell} |U| = \gamma|U|.
\end{align*}
Hence, $(\sigma_U, \sigma_V)$ covers less than $\gamma|U|$ left-vertices as desired.

\medskip

\noindent{\bf Running Time Lower Bound.}
Our construction gives a \MaxCov instance $\Gamma$ with $|V| = q$ and
$|\Sigma_V|=|\Sigma_{\wtV}|^{n/q}=2^{\Theta(n/q)}$; note also that $|U| = m^\ell$ and $|\Sigma_U| = |\Sigma_{\wtU}|^{\ell} = (1/\gamma)^{O(1)}$. Assume that Gap-ETH holds and let $\delta_0$ be the constant from \Cref{thm:Gap-ETH restated}. Moreover, let $\delta$ be any positive constant such that $\delta < \delta_0 / c$ where $c$ is the constant such that $|\Sigma_V| \leq 2^{cm/q}$.

Suppose for the sake of contradiction that, for some $q \geq \rho$ and $1 \geq \gamma > 0$, there is an algorithm that distinguishes whether $\MaxCov(\Gamma) = |U|$ or $\MaxCov(\Gamma) < \gamma|U|$ in $O_{q, \gamma}(|\Gamma|^{\delta q})$ time. Observe that, in our reduction, $|U|, |V|, |\Sigma_U| = |\Sigma_V|^{o(1)}$. Hence, the running time of the algorithm on input $\Gamma$ is $O_{q, \gamma}(|\Sigma_V|^{\delta q(1 + o(1))}) \leq O_{q, \gamma}(|\Sigma_V|^{\delta_0 q / c}) \leq O(2^{\delta_0 m})$ where the first inequality comes from our choice of $\delta$ and the second comes from $|\Sigma_V| \leq 2^{cm/q}$. Thanks to the completeness and soundness of the reduction, this algorithm can also distinguish whether $\MaxCov(\wtGamma) = |\wtU|$ or $\MaxCov(\wtGamma) < (1 - \varepsilon)|\wtU|$ in time $O(2^{\delta_0 m})$. From \Cref{thm:Gap-ETH restated}, this is a contradiction.
\end{proof}


\subsection{Proof of \Cref{thm:MinLab}}

We conclude this section with the proof of \Cref{thm:MinLab}. The proof proceeds simply by showing that, if an algorithm can distinguish between the two cases in the statement of \Cref{thm:MinLab}, it can also distinguish between the two cases in \Cref{thm:small V} (with an appropriate value of $\gamma$).

\begin{proof}[Proof of \Cref{thm:MinLab}]
Consider the label cover instance $\Gamma = (G = (U, V, E), \Sigma_U, \Sigma_V, \Pi)$ given by~\Cref{thm:small V} when $\gamma = (r/q)^{-q}$. Let us assume w.l.o.g. that there is no isolated vertex in $G$.

\medskip
\noindent{\bf Completeness.} If $\MaxCov(\Gamma) = |U|$, then there is a labeling $\sigma_U: U \to \Sigma_U$ and $\sigma_V: V \to \Sigma_V$ that covers every edge; this also induces a multi-labeling that covers every edge. Hence, $\MinLab(\Gamma) = |V|$.

\medskip
\noindent{\bf Soundness.}
We will prove by contrapositive; suppose that $\MinLab(\Gamma) \leq r$. This implies that there exists a multi-labeling $\sigma_U: U \to \Sigma_U$ and $\sigma_V: V \to 2^{\Sigma_V}$ such that $\sum_{v \in V} |\sigma_V(v)| \leq r$ and every vertex is covered. Since there is no isolated vertex in $G$, $\sigma_V(v) \ne \emptyset$ for all $v \in V$.

Consider $\sigma_V^{\rand}: V \to \Sigma_V$ sampled randomly by, for each $v \in V$, independently pick a random element of $\sigma_V(v)$ and let $\sigma_V^{\rand}(v)$ be this element. Let us consider the expected number of $u \in U$ that are covered by the labeling $(\sigma_U, \sigma_V^{\rand})$. From linearity of expectation, we can write this as
\begin{align*}
\E_{\sigma_V^{\rand}} |\{u \in U \mid (\sigma_U, \sigma_V^{\rand}) \text{ covers } u\}|
&= \sum_{u \in U} \Pr_{\sigma_V^{\rand}}\left[(\sigma_U, \sigma_V^{\rand}) \text{ covers } u\right] \\
&= \sum_{u \in U} \prod_{v \in N(u)} \Pr\left[(\sigma_U(u), \sigma_V^{\rand}(v)) \in \Pi_{uv}\right] \\
&\geq \sum_{u \in U} \prod_{v \in N(u)} |\sigma_V(v)|^{-1} \\
&\geq \sum_{u \in U} \prod_{v \in V} |\sigma_V(v)|^{-1} \\
(\text{From AM-GM inequality}) &\geq \sum_{u \in U} \left(\frac{1}{q} \sum_{v \in V} |\sigma_V(v)|\right)^{-q} \\
&\geq \sum_{u \in U} |U| (r/q)^{-q} \\
&= \gamma |U|.
\end{align*}
where the first inequality comes from the fact that there exists $\beta \in \sigma_V(v)$ such that $(\sigma_U(u), \beta) \in \Pi_{uv}$. This implies that $\MaxCov(\Gamma) \geq \gamma|U|$, which concludes our proof.
\end{proof}

%% file: combopt.tex
\section{Hardness for Combinatorial Problems}
\label{sec:hardness-combopt}

\subsection{Maximum Clique}

Recall that, for any graph $G$, $\Clique(G)$ denotes the maximum size of any clique in $G$.  
Observe that we can check if there is a clique of size $r$ by 
checking if any subset of $r$ vertices forms a clique, 
and there are ${|V(G)| \choose r}=O(|V(G)|^r)$ possible such subsets. 
We show that this is essentially the best we can do
even when we are given a promise that a clique of size 
$q \gg r$ exists:

\begin{theorem}\label{thm:clique}
	Assuming Gap-ETH, there exist constants $\delta, r_0 > 0$ such that, for any positive integers $q\geq r\geq r_0$, no algorithm can take a graph $G$ and distinguish between the following cases in $O_{q, r}(|V(G)|^{\delta r})$ time: 
	\begin{itemize}
		\item 	$\Clique(G)\geq q$ and  
		\item $\Clique(G)< r$.
	\end{itemize}
\end{theorem}

%

The above theorem simply follows from plugging the FGLSS reduction below to \Cref{thm:small r}.

\begin{theorem}[\cite{FGLSS96}] \label{thm:fglss}
	Given a label cover instance $\Gamma=(G = (U, V, E), \Sigma_U, \Sigma_V, \Pi)$ with projection property as in \Cref{sec:label cover}, there is a reduction that produces a graph $H_{\Gamma}$ such that $|V(H_{\Gamma})| = |U||\Sigma_U|$ and $\Clique(H_{\Gamma})  = \MaxCov(\Gamma)$.  The reduction takes $O(|V(H_{\Gamma}))|^2|V|)$ time. 
\end{theorem}

For clarity, we would like to note that, while the original graph defined in~\cite{FGLSS96} is for multi-prover interactive proof, analogous graphs can be constructed for CSPs and label cover instances as well. In particular, in our case, the graph can be defined as follows:
\begin{itemize}
\item The vertex set $V(H_\Gamma)$ is simply $U \times \Sigma_U$.
\item 
There is an edge between two vertices $(u, \alpha), (u', \alpha') \in V(H_\Gamma)$ if and only if, $\Pi_{uv}(\alpha) = \Pi_{u'v}(\alpha')$ (i.e., recall that we have a projection constraint, so we can represent the constraint $\Pi_{uv}$ as a function $\Pi_{uv}: \Sigma_U \rightarrow \Sigma_V$.) 
\end{itemize}

\begin{proof}[Proof of \Cref{thm:clique}] 
Assume that Gap-ETH holds and let $\delta, \rho$ be the constants from \Cref{thm:small r}. Let $r_0 = \max\{\rho, 2/\delta\}$. Suppose for the sake of contradiction that, for some $q \geq r \geq r_0$, there is an algorithm $\A$ that distinguishes between $\Clique(G) \geq q$ and $\Clique(G) < r$ in $O_{q, r}(|V(G)|^{\delta r})$ time.

Given a label cover instance $\Gamma$ with projection property, we can use $\A$ to distinguish whether $\MaxCov(\Gamma) \geq q$ or $\MaxCov(\Gamma) < r$ as follows. First, we run the FGLSS reduction to produce a graph $H_\Gamma$ and we then use $\A$ to decide whether $\Clique(H_\Gamma) \geq q$ or $\Clique(H_\Gamma) < r$. From $\Clique(H_{\Gamma})  = \MaxCov(\Gamma)$, this indeed correctly distinguishes between $\MaxCov(\Gamma) \geq q$ and $\MaxCov(\Gamma) < r$; moreover, the running time of the algorithm is $O_{q, r}(|V(H_\Gamma)|^{\delta r}) + O(|V(H_{\Gamma}))|^2|V|) \leq O_{q, r}(|\Gamma|^{\delta r})$ where the term $O(|V(H_{\Gamma}))|^2|V|)$ comes from the running time used to produce $H_\Gamma$. From \Cref{thm:small r}, this is a contradiction, which concludes our proof.
\end{proof}

As a corollary of \Cref{thm:clique}, we immediately arrive at FPT inapproximability of Maximum Independent Set and Maximum Clique.

\begin{corollary}[Clique is inherently enumerative]
Assuming Gap-ETH, Maximum Clique and Maximum Independent Set are inherently enumerative and thus FPT inapproximable.
\end{corollary}

\subsection{Set Cover, Dominating Set, and Hitting Set} \label{subsec:domset}
For convenience, we will be working with the Set Cover problem, which is computationally equivalent to Dominating Set (see Appendix~\ref{app:trivial-eq}).

Let $\uset$ be a ground set (or a universe).
A set system $\sset$ over $\uset$ is a collection of subsets $\sset = \{S_1,\ldots, S_m\}$  where $S_i \subseteq \uset$ for all $i \in [m]$. 
We say that $\sset' \subseteq \sset$ is a feasible set cover of $(\uset, \sset)$ if $\bigcup_{X \in \sset'} X = \uset$. 
In the Set Cover problem ($\SetCov$), we are given such a set system $(\uset, \sset)$ and we are interested in finding a set cover $\sset'$ with minimum cardinality $|\sset'|$.  
Let $\SetCov(\uset, \sset)$ denote the value of the optimal set cover for $(\uset, \sset)$. 

Note that for any set cover instance  $(\uset, \sset)$, checking whether there is a set cover of size at most  $q$ can be done in $\Ostar(|\sset|^q)$ time by enumerating all ${|\sset| \choose q}$ subsets of $\sset$ of size $q$. We show that this is more or less the best we can do: Even when the algorithm is promised the existence of a set cover of size $q$ (for some constant $q$), it cannot find a set cover of size $f(q)$ for any computable function $f$ in time $O_q(|\sset| |\uset|)^{\delta q}$ for some constant $\delta > 0$ independent of $q$ and $f$.

\subsubsection{Results} 
Our main technical contribution in this section is summarized in the following theorem: 

\begin{theorem}
\label{thm:setcov-reduction}
There is a reduction that on input $\Gamma= (G=(U,V,E),\Sigma_U, \Sigma_V, \Pi)$ of $\MinLab$ instance, produces a set cover instance $(\uset,\sset)$ such that 

\begin{itemize}

\item $\MinLab(\Gamma) = \SetCov(\uset,\sset)$  

\item $|\uset| = |U| |V|^{|\Sigma_U|}$ and $|\sset| = |V| |\Sigma_V|$ 

\item The reductions runs in time $poly(|\uset|, |\sset|)$ 
\end{itemize}

\end{theorem}

We defer the proof of this theorem to \Cref{subsec:domset}. For now, let us demonstrate that, by combining \Cref{thm:setcov-reduction} and \Cref{thm:small V}, we can derive hardness of approximation of \SetCov:

\begin{theorem}
		\label{thm: set cover hardness} 
	Assuming Gap-ETH, there exist universal constants $\delta, q_0 > 0$ such that, for any positive integers $r \geq q \geq q_0$,  no algorithm can take a set cover instance $(\uset, \sset)$, and distinguish between the following cases in $O_{q, r}((|\sset| |\uset|)^{\delta q})$ time: 
	\begin{itemize}
	\item $\SetCov(\uset, \sset) \leq q$. 
	\item $\SetCov(\uset,\sset) > r$. 
    \end{itemize}
\end{theorem}

\begin{proof}
Assume that Gap-ETH holds and let $\delta, \rho$ be the constants from \Cref{thm:MinLab}. Let $q_0 = \max\{\rho, c/\delta\}$ where $c$ is the constant such that the running time of the reduction in \Cref{thm:setcov-reduction} is $O((|\uset||\sset|)^c)$. Suppose for the sake of contradiction that, for some $r \geq q \geq q_0$, there is an algorithm $\A$ that distinguishes between $\SetCov(\uset, \sset) \leq q$ and $\SetCov(\uset, \sset) > r$ in $O_{q, r}((|\sset| |\uset|)^{\delta q})$ time.

Given a label cover instance $\Gamma$ where $|V|, |\Sigma_U| = O_{q, r}(1)$, we can use $\A$ to distinguish whether $\MinLab(\Gamma) \leq q$ or $\MinLab(\Gamma) > r$ as follows. First, we run the reduction from \Cref{thm:setcov-reduction} to produce a \SetCov instance $(\uset, \sset)$ and we then use $\A$ to decide whether $\SetCov(\uset, \sset) \leq q$ or $\SetCov(\uset, \sset) > r$. From $\SetCov(\uset, \sset)  = \MinLab(\Gamma)$, this indeed correctly distinguishes between $\MinLab(\Gamma) \leq q$ and $\MinLab(\Gamma) > r$; moreover, the running time of the algorithm is $O_{q, r}((|\uset||\sset|)^{\delta q}) + O((|\uset||\sset|)^c) \leq O_{q, r}(|\Gamma|^{\delta q})$ where the term $O((|\uset||\sset|)^c)$ comes from the running time used to produce $(\uset, \sset)$. From \Cref{thm:MinLab}, this is a contradiction, which concludes our proof.
\end{proof}

As a corollary of \Cref{thm: set cover hardness}, we immediately arrive at FPT inapproximability of Set cover, Dominating set and Hitting set.

\begin{corollary}
Assuming Gap-ETH, Set cover, Dominating set and Hitting set are inherently enumerative and thus FPT inapproximable. 
\end{corollary}

\subsubsection{Proof of Theorem~\ref{thm:setcov-reduction}} \label{subsec:domset}

Our construction is based on a standard hypercube set system, as used by Feige~\cite{Feige98} in proving the hardness of $k$-Maximum Coverage.  
We explain it here for completeness. 

\paragraph{Hypercube set system:}
Let $z, k \in {\mathbb N}$ be parameters. The hypercube set system $H(z,k)$ is a set system $(\uset, \sset)$ with the ground set $\uset= [z]^k$. We view each element of $\uset$ as a length-$k$ vector $\vec{x}$ where each coordinate assumes a value in $[z]$.  
There is a collection of {\em canonical sets} $\sset = \{X_{i,a}\}_{i \in [z], a \in [k]}$ defined as 
\[X_{i,a} = \{\vec{x}: \vec{x}_a = i \} \] 
In other words, each set $X_{i,a}$ contains the vectors whose $a^{th}$ coordinate is $i$.
A nice property of this set system is that, it can only be covered completely if all canonical sets corresponding to some $a^{th}$ coordinate are chosen. 

\begin{proposition}
\label{prop:hypercube} 
Consider any sub-collection $\sset' \subseteq \sset$. 
We have $\bigcup \sset' = \uset$ if and only if there is a value $a \in [k]$ for which $X_{1,a}, X_{2,a},\ldots, X_{z,a} \in \sset'$.  
\end{proposition}
\begin{proof}
The if part is obvious. 
For the ``only if'' part, assume that for each $a \in [k]$, there is a value $i_a \in [z]$ for which $X_{i_a, a}$ is not in $\sset'$. Define vector $\vec{x}$ by $\vec{x}_a = i_a$. Notice that $\vec{x}$ does not belong to any set in $\sset'$ (By definition, if $X_{i',a'}$ contains $\vec{x}$, then it must be the case that $\vec{x}_{a'} = i' = i_{a'}$.)   
\end{proof}

\paragraph{The construction:}
Our reduction starts from the $\MinLab$ instance $\Gamma = (G, \Sigma_U, \Sigma_V, \Pi)$. 
We will create the set system $\iset= (\uset, \sset)$.  
We make $|U|$ different copies of the hypercube set system: For each vertex $u\in U$, we have the hypercube set system $(\uset^u, \sset^u) = H(N_G(u), \Sigma_U)$, i.e., the ground set $\uset^u$ is a copy of $N_G(u)^{\Sigma_U}$  and $\sset^u$ contains $|N_G(u)| |\Sigma_U|$ ``virtual'' sets, that we call $\{S^u_{v,a}\}_{v \in N_G(u), a \in \Sigma_U}$ where each such set corresponds to a canonical set of the hypercube. 
We remark that these virtual sets are not the eligible sets in our instance $\iset$. 
For each vertex $v \in V$, for each label $b \in \Sigma_V$, we define a set 
\[S_{v,b} = \bigcup_{u \in N_G(v), (a,b) \in \Pi_{uv}}  S^u_{v,a}\]  
The set system $(\uset, \sset)$ in our instance is simply: 
$$\uset = \bigcup_{u \in U} \uset^u \hspace{0.2in} \mbox{ and } \hspace{0.2in} \sset = \{S_{v,b}: v \in V, b \in \Sigma_V\} $$

Notice that the number of sets is $|V||\Sigma_V|$ and the number of elements in the ground set is $|\uset| = |U| |V|^{|\Sigma_U|}$.
This completes the description of our instance. 

\paragraph{Analysis:} We argue that the optimal value of $\Gamma$ is equal to the optimal of $(\uset,\sset)$. 

First, we will show that $\MinLab(\Gamma) \leq \SetCov(\uset, \sset)$. Let $(\sigma_U, \hat{\sigma}_V)$ be a feasible $\MinLab$ cover for $\Gamma$ (recall that $\hat{\sigma}_V$ is a multi-labeling, while $\sigma_U$ is a labeling.) 
For each $v \in V$, the \SetCov solution chooses the set $S_{v,b}$ for all $b \in \hat{\sigma}_V(v)$. Denote this solution by $\sset' \subseteq \sset$. 
The total number of sets chosen is exactly $\sum_v |\hat{\sigma}(v)|$, exactly matching the cost of $\MinLab(\Gamma)$. 
We argue that this is a feasible set cover: For each $u$, the fact that $u$ is covered by $(\sigma_U, \hat{\sigma}_V)$ implies that, for all $v \in N_G(u)$, there is a label $b_v \in \hat{\sigma}_V(v)$ such that $(\sigma_U(u), b_v) \in \Pi_{uv}$. Notice that $S^u_{v,\sigma_U(u)} \subseteq S_{v, b_v} \in \sset'$ for every $v \in N_G(u)$, so we have $$\bigcup_{S \in \sset'} S \supseteq \bigcup_{v \in N_G(u)} S_{v, b_v} \supseteq \bigcup_{v \in N_G(u)} S^u_{v, \sigma_U(u)} = \uset^u$$ where the last equality comes from \Cref{prop:hypercube}. In other words, $\sset'$ covers all elements in $\uset^u$. Hence, $\sset'$ is indeed a valid $\SetCov$ solution for $(\uset, \sset)$.


To prove the converse, consider a collection of sets $\{S_{v,b}\}_{(v,b) \in \Lambda}$ that covers the whole universe $\uset$. 
We define the (multi-)labeling  $\hat{\sigma}_V: V \rightarrow 2^{\Sigma_V}$ where $\hat{\sigma}_V(v) = \{b: (v,b) \in \Lambda\}$ for each $v \in V$. 
Clearly, $\sum_{v \in V} |\hat{\sigma}_V(v)| = |\Lambda|$, so the cost of $\hat{\sigma}_V$ as a solution for $\MinLab$ is exactly the cost of \SetCov.  
We verify that all left vertices $u \in U$ of $\Gamma$ are covered (and along the way will define $\Sigma_U(u)$ for all $u \in U$.)
Consider each vertex $u \in U$. 
The fact that the ground elements in $\uset^u$ are covered implies that (from Proposition~\ref{prop:hypercube}) there is a label $a_u \in \Sigma_U$ where all virtual sets $\{S^u_{v,a_u}\}_{v \in N_G(u)}$ are included in the solution. 
Therefore, for each $v \in N_G(u)$, there must be a label $b_v \in \hat{\sigma}_V(v)$ such that $a_u b_v \in \Pi_{uv}$.  
We simply define $\sigma_U(u) = a_u$. 
Therefore, the vertex $u$ is covered by the assignment $(\sigma_U, \hat{\sigma}_V)$. 

\subsection{Maximum Induced Subgraph with Hereditary Properties} 

In this section, we prove the hardness of maximum induced subgraphs with hereditary property.
Let $\Pi$ be a graph property.  
We say that a subset $S \subseteq V(G)$ has property $\Pi$ if $G[S] \in \Pi$. 
Denote by $A_{\Pi}(G)$ the maximum cardinality of a set $S$ that has property $\Pi$.

Khot and Raman~\cite{KhotR00} proved a dichotomy theorem for the problem; if $\Pi$ contains all independent sets but not all cliques or if $\Pi$ contains all cliques but not all independent sets, then the problem is W[1]-hard. For all other $\Pi$'s, the problem is in FPT.
We will extend Khot and Raman's dichotomy theorem to hold even for FPT approximation as stated more precisely below.

\begin{theorem}
Let $\Pi$ be any hereditary property. 
\begin{itemize}
\item If $\Pi$ contains all independent sets but not all cliques or vice versa, then computing $A_{\Pi}(G)$ is weakly inherently enumerative (and therefore totally FPT inapproximable). 

\item Otherwise, $A_{\Pi}(G)$ can be computed exactly in FPT. 
\end{itemize}
\end{theorem}


Surprisingly, the fact that there is a gap in the optimum of our starting point helps make our reduction simpler than that of Khot and Raman.
For convenience, let us focus only on properties $\Pi$'s which contain all independent sets but not all cliques. The other case can be proved analogously. The main technical result is summarized in the following lemma. 

\begin{theorem}
\label{thm:hereditary} 
Let $\Pi$ be any graph property that contains all independent sets but not all cliques. Then there is a function $g_\Pi = \omega(1)$ such that the following holds: 
\begin{itemize}
\item If $\alpha(G) \geq q$, then $A_{\Pi}(G) \geq q$.
\item If $A_{\Pi}(G) \geq r$, then $\alpha(G) \geq g_\Pi(r)$.
\end{itemize}
\end{theorem}
\begin{proof}
Since $\Pi$ contains all independent set, when $\alpha(G) \geq q$, we always have $A_{\Pi}(G) \geq q$. 

Now, to prove the converse, let $g_\Pi(r)$ denote $\max_{H \in \Pi, |V(H)| = r} \alpha(H)$. If $A_{\Pi}(G) = r$, then there exists a subset $S \subseteq V(G)$ of size $r$ that has property $\Pi$; from the definition of $g_\Pi$, $\alpha(H) \geq g_\Pi(r)$, which implies that $\alpha(G) \geq g_\Pi(r)$ as well. Hence, we are only left to show that $g_\Pi = \omega(1)$.

To show that this is the case, recall the Ramsey theorem.

\begin{theorem}[Ramsey's Theorem]
For any $s, t \geq 1$, there is an integer $R(s, t)$ s.t. every graph on $R(s, t)$ vertices contains either a $s$-clique or a $t$-independent set. Moreover, $R(s, t) \leq \binom{s + t - 2}{s - 1}$.
\end{theorem}

Recall that, from our assumption of $\Pi$, there exists a fixed integer $s_{\Pi}$ such that $\Pi$ does not contain an $s_{\Pi}$-clique. Hence, from Ramsey's Theorem, $g_\Pi(r) \geq \max\{t \mid R(s_{\Pi}, t) \leq r\}$. In particular, this implies that $g_\Pi(r) \geq \Omega_{s_{\Pi}}(r^{1/(s_{\Pi - 1})})$. Hence, $\lim_{r \infty} g_\Pi(r) = \infty$ (i.e. $g_\Pi = \omega(1)$) as desired.
\end{proof}

In other words, the identical transformation $G \mapsto G$ is a $(q, g_\Pi(r))$-FPT gap reduction from \Clique to Maximum Induced Subgraph with property $\Pi$. Hence, by applying~\Cref{prop:gapred-enum}, we immediately arrive at the following corollary.

\begin{corollary}
Assuming Gap-ETH, for any property $\Pi$ that contains all independent sets but not all cliques (or vice versa), Maximum Induced Subgraph with property $\Pi$ is $\Omega(g_\Pi)$-weakly inherently enumerative where $g_\Pi$ is the function from Theorem~\ref{thm:hereditary}.
\end{corollary}

We remark here that, for some properties, $g_\Pi$ can be much larger than the bound given by the Ramsey's Theorem; for instance, if $\Pi$ is planarity, then the Ramsey's Theorem only gives $g_\Pi(r) = \Omega(r^{1/5})$ but it is easy to see that, for planar graphs, there always exist an independent set of linear size and $g_\Pi(r)$ is hence as large as $\Omega(r)$.

\subsection{Maximum Balanced Biclique, Maximum Induced Matching on Bipartite Graphs and Densest $k$-Subgraph}

We next prove FPT inapproximability for the Maximum Balanced Biclique, Maximum Induced Matching on Bipartite Graphs and Densest $k$-Subgraph. Unlike the previous proofs, we will not reduce from any label cover problem; the starting point for the results in this section will instead be a recent construction of Manurangsi for ETH-hardness of Densest $k$-Subgraph~\cite{Man17}. By interpreting this construction in a different perspective, we can modify it in such a way that we arrive at a stronger form of inherently enumerative hardness for \Clique. More specifically, the main theorem of this section is the following theorem, which is a stronger form of \Cref{thm:clique} in that the soundness not only rules out cliques, but also rules out bicliques as well.

\begin{theorem} \label{thm:cliquevbiclique}
Assuming Gap-ETH, there exist constants $\delta, \rho > 0$ such that, for any positive integers $q \geq r \geq \rho$, no algorithm can take a  graph $G$ and distinguish between the following cases in $O_{q, r}(|V(G)|^{\delta\sqrt{r}})$ time:
\begin{itemize}
\item $\Clique(G) \geq q$.
\item $\BiClique(G) < r$.
\end{itemize}
\end{theorem}

The weakly inherently enumerativeness (and therefore totally FPT inapproximability) of Maximum Balanced Biclique and Maximum Induced Matching on Bipartite Graphs follows easily from Theorem~\ref{thm:cliquevbiclique}. We will show these results in the subsequent subsections; for now, let us turn our attention to the proof of the theorem. 

The main theorem of this section can be stated as follows.

\begin{theorem} \label{thm:mbb-fpt-inapprox}
For any $d, \varepsilon > 0$, there is a constant $\gamma = \gamma(d, \varepsilon) > 0$ such that there exists a (randomized) reduction that takes in a parameter $r$ and a 3-\SAT instance $\phi$ with $n$ variables and $m$ clauses where each variable appears in at most $d$ constraints and produces a graph $G_{\phi, r} = (V_{\phi, r}, E_{\phi, r})$ such that, for any sufficiently large $r$ (depending only on $d, \varepsilon$ but not $n$), the following properties hold with high probability:
\begin{itemize}
\item (Size) $N := |V_{\phi, r}| \leq 2^{O_{d, \varepsilon}(n/\sqrt{r})}$.
\item (Completeness) if $\SAT(\phi) = m$, then $\Clique(G_{\phi, r}) \geq N^{\gamma/\sqrt{r}}$.
\item (Soundness) if $\SAT(\phi) \leq (1 - \varepsilon)m$, then $\BiClique(G_{\phi, r}) < r$.
\end{itemize}
\end{theorem}

It is not hard to see that, in the Gap-ETH assumption, we can, without loss of generality, assume that each variable appears in only a bounded number of clauses (See~\cite[p.~21]{ManR16}). Hence, Theorem~\ref{thm:mbb-fpt-inapprox} together with Gap-ETH implies Theorem~\ref{thm:cliquevbiclique}.

As mentioned earlier, our result builds upon an intermediate lemma used to prove the hardness of approximating Densest $k$-Subgraph in~\cite{Man17}. Due to this, it will be easier to describe our reduction in terms of the reduction from~\cite{Man17}; in this regard, our reduction can be viewed as vertex subsampling (with appropriate probability) of the graph produced by the reduction from~\cite{Man17}. The reduction is described formally in Figure~\ref{fig:red-cliquevbiclique}. Note that the two parameters $\ell$ and $p$ will be chosen as $\Theta_{d, \varepsilon}(n/\sqrt{r})$ and $2^{\Theta_{d, \varepsilon}(\ell^2/n)}/\binom{n}{\ell}$ respectively where the constants in $\Theta_{d, \varepsilon}(\cdot)$ will be selected based on the parameters from the intermediate lemma from~\cite{Man17}.

\begin{figure}[h!]
\begin{framed}
Input: a 3-\SAT instance $\phi$ and parameters $p \in (0, 1)$ and $\ell \in \mathbb{N}$ such that $\ell \leq n$. \\
Output: a graph $G_{\phi, \ell, p} = (V_{\phi, \ell, p}, E_{\phi, \ell, p})$. \\
The graph $G_{\phi, \ell, p}$ is generated as follows.
\begin{itemize}
\item First, we create a graph $\wtG_{\phi, \ell} = (\wtV_{\phi, \ell}, \wtE_{\phi, \ell})$ as constructed in~\cite{Man17}. More specifically, the vertex set $\wtV_{\phi, \ell}$ and the edge set $\wtE_{\phi, \ell}$ are defined as follows.
\begin{itemize}
\item The vertex set $\wtV_{\phi, \ell}$ consists of all partial assignments of $\ell$ variables, i.e., $\wtV_{\phi, \ell} := \{\sigma: S \rightarrow \{0, 1\} \mid S \in \binom{\cX}{\ell}\}$ where $\cX$ is the set of all variables in $\phi$.
\item There exists an edge between two vertices $\sigma_1: S_1 \rightarrow \{0, 1\}$ and $\sigma_2: S_2 \rightarrow \{0, 1\} \in \wtV_{\phi, \ell}$ if and only if (1) they are consistent (i.e., $\sigma_1(S_1 \cap S_2) = \sigma_2(S_1 \cap S_2)$) and (2) the partial assignment induced by $\sigma_1, \sigma_2$ does not violate any constraint (i.e., every constraint that lies entirely inside $S_1 \cup S_2$ is satisfied by the partial assignment induced by $\sigma_1, \sigma_2$).
\end{itemize}
\item Our graph $G_{\phi, \ell, p} = (V_{\phi, \ell, p}, E_{\phi, \ell, p})$ can then be easily generated as follows.
\begin{itemize}
\item Let $V_{\phi, \ell, p}$ be a random subset of $\wtV_{\phi, \ell}$ such that each vertex $v \in \wtV_{\phi, \ell}$ is included independently and randomly in $V_{\phi, \ell, p}$ with probability $p$.
\item We connect $u, v \in V_{\phi, \ell, p}$ if and only if $(u, v) \in \wtE_{\phi, \ell}$.
\end{itemize}
\end{itemize}
\end{framed}
\caption{The Reduction from Gap-3\SAT to Maximum Balanced Biclique}
\label{fig:red-cliquevbiclique}
\end{figure}

The main lemma of~\cite{Man17} is stated below. Roughly speaking, when $\SAT(\phi) \leq (1 - \varepsilon)m$, the lemma gives an upper bound on the number of occurrences of $K_{t, t}$ for every $t > 0$. When $p$ and $t$ are chosen appropriately, this implies that w.h.p. there is no $t$-biclique in our subsampled graph. Note that the size and completeness properties are obvious from the construction while the exact statement of the soundness can be found in the proof of Theorem 8 of~\cite{Man17}. 

\begin{lemma}[\cite{Man17}] \label{lem:dks-reduction} 
Let $d, \varepsilon, \phi, n, m, \ell$ be as in Theorem~\ref{thm:mbb-fpt-inapprox} and Figure~\ref{fig:red-cliquevbiclique}. There is a constant $\delta, \lambda > 0$ depending only on $d, \varepsilon$ such that, for any sufficiently large $n$, the graph $G_{\phi, \ell} = (V_{\phi, \ell}, E_{\phi, \ell})$ described in Figure~\ref{fig:red-cliquevbiclique} has the following properties
\begin{itemize}
\item (Size) $|V_{\phi, \ell}| = \binom{n}{\ell}2^\ell$.
\item (Completeness) if $\SAT(\phi) = m$, $\wtG_{\phi, \ell}$ contains a $\binom{n}{\ell}$-clique.
\item (Soundness) if $\SAT(\phi) \leq (1 - \varepsilon)m$, then $\wtG_{\phi, \ell}$ contains at most $2^{4n}(2^{-\lambda \ell^2 / n} \binom{n}{\ell})^{2t}$ occurrences\footnote{We say that $S, T \subseteq V_{\phi, \ell}$ is an occurrence of $K_{t, t}$ if $|S| = |T| = t$, $S \cap T = \emptyset$ and, for every $s \in S, t \in T$, there is an edge between $s$ and $t$ in $G_{\phi, \ell}$. The number of occurrences of $K_{t, t}$ of $G_{\phi, \ell}$ is simply the number of such pairs $(S, T)$'s.} of $K_{t, t}$ for any $t > 0$.
\end{itemize}
\end{lemma}

Theorem~\ref{thm:mbb-fpt-inapprox} follows rather easily from the above lemma by choosing appropriate $\ell$ and $p$.

\begin{proof}[Proof of Theorem~\ref{thm:mbb-fpt-inapprox}]
We let $G_{\phi, r} = G_{\phi, \ell, p}$ from the reduction in Figure~\ref{fig:red-cliquevbiclique} with parameters $\ell = \frac{4n}{\sqrt{\lambda r}}$ and $p = 2^{\frac{\lambda \ell^2}{2n}} / \binom{n}{\ell}$. For convenience, we assume without loss of generality that $\lambda < 1$.

\paragraph{Size.} Since each vertex in $V_{\phi, \ell}$ is included that $V_{\phi, \ell, p}$ independently with probability $p$, we have
$\E[|V_{\phi, \ell, p}|] = p|V_{\phi, \ell}| = 2^{\ell + \frac{\lambda \ell^2}{2n}} \leq 2^{2\ell}$. Hence, from Chernoff bound, $|V_{\phi, \ell, p}| \leq 2^{10\ell} = 2^{\Omega_{d, \varepsilon}(n/\sqrt{r})}$ w.h.p.

\paragraph{Completeness.} Suppose that $\phi$ is satisfiable. Let $C$ be the clique of size $\binom{n}{\ell}$ in $\wtG_{\phi, \ell}$, which is guaranteed to exist by Lemma~\ref{lem:dks-reduction}. From how $G_{\phi, \ell, p}$ is defined, $C \cap V_{\phi, \ell, p}$ induces a clique in $G_{\phi, \ell, p}$. Moreover, $\E[|C \cap V_{\phi, \ell, p}|] = p|C| = 2^{\frac{\lambda \ell^2}{2n}}$. Again, from Chernoff bound, $\Clique(G_{\phi, \ell, p}) \geq 2^{\frac{\lambda \ell^2}{2n}}$ w.h.p. Combined with the above bound on $N$, $\Clique(G_{\phi, \ell, p}) \geq N^{\gamma/\sqrt{r}}$ w.h.p. when $\gamma := \sqrt{\lambda}/20 = O_{d, \varepsilon}(1)$.

\paragraph{Soundness.} Suppose that $\SAT(\phi) \leq (1 - \varepsilon)m$. Consider any subsets $S, T \subseteq \wtV_{\phi, \ell}$ that is an occurrence of $K_{r, r}$ in $\wtG_{\phi, \ell}$. From how $G_{\phi, \ell, p}$ is defined, $\BiClique(G_{\phi, \ell, p}) \geq r$ if and only if, for at least one such pair $(S, T)$, $S \cup T \subseteq V_{\phi, \ell, p}$. The probability of this event is bounded above by
\begin{align*}
\sum_{S, T \subseteq \wtV_{\phi, \ell} \atop S, T \text{ is an occurrence of } K_{r, r} \text{ in } \wtG_{\phi, \ell}} \Pr[S, T \subseteq V_{\phi, \ell, p}] &\leq 2^{4n}\left(2^{-\lambda \ell^2 / n} \binom{n}{\ell}\right)^{2r} \cdot p^{2r} \\
&= 2^{4n} \left(2^{-\frac{\lambda \ell^2}{2n}}\right)^{2r} \\
&= o(1).
\end{align*}
where the first inequality comes from the bound in the soundness of Lemma~\ref{lem:dks-reduction} and the fact that the sampling of each vertex is done independently.

As a result, the subsampled graph $G_{\phi, \ell, p}$ is $K_{r, r}$-free with high probability as desired.
\end{proof}

\subsubsection{Maximum Balanced Biclique}

We now give a simple reduction from the ``\Clique vs \BiClique'' problem (from~\Cref{thm:cliquevbiclique}) to the Maximum Balanced Biclique problem, which yields FPT inapproximability of the latter.

\begin{lemma}
For any graph $G = (V, E)$, let $B_e[G] = (V_{B_e[G]}, E_{B_e[G]})$ be the bipartite graph whose vertex set is $V_{B_e[G]} := V \times [2]$ and two vertices $(u, i), (v, j)$ are connected by an edge if and only if $(u, v) \in E$ or $u = v$, and $i \ne j$. Then the following properties hold for any graph $G$.
\begin{itemize}
\item $\BiClique(B_e[G]) \geq \Clique(G)$.
\item $\BiClique(B_e[G]) \leq 2\BiClique(G) + 1$.
\end{itemize}
\end{lemma}

\begin{proof}
It is easy to see that $\BiClique(B_e[G]) \geq \Clique(G)$ since, for any $C \subseteq V$ that induces a clique in $G$, $C \times [2] \subseteq V_{B_e[G]}$ induces a $|C|$-biclique in $B_e[G]$.

To see that $\BiClique(B_e[G]) \leq 2\BiClique(G) + 1$, consider any $S \subseteq V_{B_e[G]}$ that induces a $k$-biclique in $B_e[G]$. 
Note that $S$ can be partitioned into $S_1 = S \cap (V \times \{1\})$ and $S_2 = S \cap (V \times \{2\})$. 

Now consider the projections of $S_1$ and $S_2$ into $V(G)$, i.e.,
$T_1 = \{v: (v,1) \in S\}$ and $T_2 = \{v: (v,2) \in S\}$.
Note that $|T_1| = |T_2| = k$.
Since $S_1 \cup S_2$ induces a biclique in $B_e[G]$, we have, for every $u \in T_1$ and $v \in T_2$, either $u = v$ or $(u, v) \in E$.
Observe that if there were no former case (i.e., $T_1\cap T_2=\emptyset$), 
then we would have a $k$-biclique in $G$.
Even if $T_1 \cap T_2 \ne \emptyset$, we can still get back a $\lfloor k/2 \rfloor$-biclique of $G$ by uncrossing the sets $T_1$ and $T_2$ in a natural way by 
assigning half of the intersection to $T_1$ and the other half to $T_2$.
To be formal, we partition $T_1\cap T_2$ into roughly equal sets
$U_1$ and $U_2$ (i.e., $||U_1| - |U_2|| \leq 1$),
and we then define new sets $T'_1$ and $T'_2$ by
\[
T'_1 = (T_1\setminus T_2) \cup U_1
\mbox{ and }
T'_2 = (T_2\setminus T_1) \cup U_2.
\]
It is not hard to see that $G$ has an edge between every pair of vertices
between $T'_1, T'_2$ and that $|T'_1|, |T'_2| \geq \lfloor k/2 \rfloor$. Thus, $\BiClique(G) \geq \lfloor k/2 \rfloor \geq (k - 1)/2$. Therefore, $\BiClique(B_e[G]) \leq 2\BiClique(G) + 1$ as desired. 
\end{proof}

Thanks to the above lemma, we can conclude that the reduction $G \mapsto B_e[G]$ is a $(2q, (r + 1)/2)$-FPT gap reduction from the ``\Clique vs \BiClique'' problem to Maximum Balanced Biclique, although the former is not a well-defined optimization problem. Nevertheless, it is easy to check that a proof along the line of~\Cref{prop:gapred-enum} still works and it gives the following result:

\begin{corollary}
Assuming Gap-ETH, Maximum Balanced Biclique are $\Omega(\sqrt{r})$-weakly inherently enumerative and thus FPT inapproximable.
\end{corollary}

It is worth noting here that the Maximum Edge Biclique problem, a well-studied variant of the Maximum Balanced Biclique problem where the goal is instead to find a (not necessarily balanced) complete bipartite subgraph of a given bipartite graph that contains as many edges as possible, is in FPT; this is because the optimum is at least the maximum degree, but, when the degree is bounded above by $r$, all bicliques can be enumerated in $2^{O(r)} \poly(n)$ time.

\subsubsection{Maximum Induced Matching on Bipartite Graphs}

Next, we prove the FPT hardness of approximation for the Maximum Induced Matching problem on bipartite graphs.
Again, the proof will be a simple reduction from Theorem~\ref{thm:cliquevbiclique}.
The argument below is similar to that used in Lemma~IV.4 of \cite{ChalermsookLN-FOCS13}.
We include it here for completeness.

\begin{lemma}
For any graph $G = (V, E)$, let $B_e[\bG] = (V_{B_e[\bG]}, E_{B_e[\bG]})$ be the bipartite graph whose vertex is $V_{B_e[\bG]} := V \times [2]$ and two vertices $(u, i), (v, j)$ are connected by an edge if and only if $(u, v) \notin E$ or $u = v$, and $i \ne j$. Then, the following properties hold for any graph $G$.
\begin{itemize}
\item $\IM(B_e[\bG]) \geq \Clique(G)$.
\item $\IM(B_e[\bG]) \leq 2\BiClique(G) + 1$.
\end{itemize}
\end{lemma}

\begin{proof}
Consider any $S \subseteq V$ that induces a clique in $G$. It is obvious that $S \times [2] \subseteq V_{B_e[\bG]}$ induces a matching in $B_e[\bG]$.

Next, consider any induced matching matching $\{(u_1, v_1), \dots, (u_m, v_m)\}$ of size $m$. Assume w.l.o.g. that $u_1, \dots, u_m \in V \times \{1\}$ and $v_1, \dots, v_m \in V \times \{2\}$. Define $\pi_1: V \times [2] \to V$ to be a projection operator that projects on to the first coordinate.

Let $S_1 = \pi_1(\{u_1, \dots, u_{\lfloor m/2 \rfloor}\})$ and $S_2 = \pi_1(\{v_{\lceil m / 2 \rceil + 1}, \dots, v_m\})$. From the definition of $B_e[\bG]$ and from the fact that there is no edge between $(S_1 \times \{1\})$ and $(S_2 \times \{2\})$, it is easy to check that $S_1 \cap S_2 = \emptyset$ and, for every $u \in S_1$ and $v \in S_2$, $(u, v) \in E$. In other words, $(S_1, S_2)$ is an occurrence of $\lfloor m / 2 \rfloor$ in $G$. Hence, we can conclude that $\IM(B_e[\bG]) \leq 2 \BiClique(G) + 1$. 
\end{proof}

Similar to \BiClique, it is easy to see that the above reduction implies the following running time lower bound and FPT inapproximability for Maximum Induced Matching on Bipartite Graphs.

\begin{corollary}
Assuming Gap-ETH, Maximum Induced Matching on Bipartite Graphs are $\Omega(\sqrt{r})$-weakly inherently enumerative and thus FPT inapproximable.
\end{corollary}

\subsubsection{Densest $k$-Subgraph}

Finally, we will show FPT inapproximability result for Densest $k$-Subgraph. Alas, we are not able to show $o(k)$-ratio FPT inapproximability, which would have been optimal since the trivial algorithm gives an $O(k)$-approximation for the problem. Nonetheless, we will show an $k^{o(1)}$-factor FPT inapproximability for the problem. We note here that below we will state the result as if $k$ is the parameter; this is the same as using the optimum as the parameter, since (in the non-trivial case) the optimum is always between $\lfloor k / 2 \rfloor$ and $\binom{k}{2}$ (inclusive).

To derive our result, we resort to a well-known result in extremal combinatorics called the K\H{o}v\'{a}ri-S\'{o}s-Tur\'{a}n (KST) Theorem, which basically states that if a graph does not contain small bicliques, then it is sparse. The KST theorem is stated formally below.

\begin{theorem}[K\H{o}v\'{a}ri-S\'{o}s-Tur\'{a}n (KST) Theorem \cite{KST54}]
For every positive integer $n$ and $t \leq n$, every $K_{t,t}$-free graph on $n$ vertices has at most $O(n^{2 - 1/t})$ edges (i.e., density $O(n^{-1/t})$).
\end{theorem}

We remark here that a generalization of the KST Theorem was also a crucial ingredient in the proof of ETH-hardness of approximating Densest $k$-Subgraph in~\cite{Man17}. The situation is simpler for us here, since we can simply apply the KST Theorem to Theorem~\ref{thm:cliquevbiclique}, which yields the following theorem.

\begin{theorem} \label{thm:gapdks}
Assuming Gap-ETH, there exist a constant $\delta > 0$ and an integer $\rho > 0$ such that, for any integer $q \geq r \geq \rho$, no algorithm can take a graph $G = (V, E)$ and distinguish between the following cases in $O_{q, r}(|V|^{\delta\sqrt{r}})$ time:
\begin{itemize}
\item $\Den_{q}(G) = 1$.
\item $\Den_{q}(G) < O(q^{- r})$.
\end{itemize}
\end{theorem}

From the above theorem, it is easy to show the $k^{o(1)}$-factor FPT inapproximability of Densest $k$-Subgraph as formalized below. We note here that our result applies to a special case of Densest $k$-Subgraph in which the input graph is promised to contain a $k$-clique; this problem is sometimes referred to as \emph{Densest $k$-Subgraph with perfect completeness}~\cite{BravermanKRW17,Man17}.

\begin{lemma}
Assuming Gap-ETH, for every function $f = o(1)$ and every function $t$, there is no $t(k) \cdot n^{O(1)}$-time algorithm such that, given an integer $k$ and any graph $G = (V, E)$ on $n$ vertices that contains at least one $k$-clique, always output $S \subseteq V$ of size $k$ such that $\Den(S) \geq k^{-f(k)}$.
\end{lemma}

\begin{proof}
Suppose for the sake of contradiction that there is a $t(k) \cdot |V|^{D}$-time algorithm $\A$ that, given an integer $k$ and any graph $G = (V, E)$ that contains a $k$-clique, always outputs $S \subseteq V$ of size $k$ such that $\Den(S) \geq k^{-f(k)}$ for some function $f = o(1)$, some function $t$ and some constant $D > 0$. 

Let $r = \max\{\lceil \rho \rceil, \lceil (D / \delta)^2 \rceil\}$ where $\rho$ is the constant from~\Cref{thm:gapdks}. Note that $O(q^{-r}) = q^{O(1)/\log q - r}$. Now, since $\lim_{q \to \infty} f(q) + O(1)/\log q = 0$, there exists a sufficiently large $q$ such that the term $O(q^{-r})$ is less than $q^{-f(q)}$. In other words, $\A$ can distinguish between the two cases in~\Cref{thm:gapdks} in time $t(q) \cdot n^D = O_{q, r}(|V|^{\delta \sqrt{r}})$, which would break Gap-ETH.
\end{proof}

%% file: conclusion.tex
\section{Conclusion and Discussions}

In this paper, we prove that \Clique and \DomSet are totally FPT inapproximable. In fact, we show a stronger property that they are inherently enumerative, i.e., the best way to approximate both problems is to essentially enumerate all possibilities. Since \Clique and \DomSet are complete problems for the class $\W[1]$ and $\W[2]$ respectively, it might be possible that these two problems can be sources of FPT-inapproximability of many other problems that admit no FPT algorithms. 

We would like to also mention that there are some problems that are known
to be totally FPT-inapproximable under weaker assumptions; examples of such problems are {\em independent dominating set}  and {\em induced path}. 
The former has been shown to be FPT-inapproximable under the
assumption $\FPT \neq \W[2]$ in \cite{DowneyFMR08}. 
For the induced path problem, we show in
\Cref{apx:w1-hardness} that it is FPT-inapproximable under
the assumption $\FPT \neq \W[1]$. It would be interesting to understand whether it is possible to also base total FPT-inapproximability of \Clique and \DomSet under assumptions that are weaker than Gap-ETH, such as $\FPT \neq \W[1]$ or ETH. To this end, we note that Chen and Lin~\cite{ChenL16} showed inapproximability for \DomSet under $\FPT \neq \W[1]$ (resp., ETH), but their inapproximability ratio is only any constant (resp., $\log^{1/4 - \varepsilon}(\OPT)$); if their result could be extended to exclude $f(\OPT)$-approximation for any function $f$, then \DomSet would indeed be totally FPT-inapproximable under weaker assumptions.


Another interesting further research direction is to
study the trade-off between the running time and the approximation ratio of 
problems that are known to be FPT-approximable or admit FPT 
(exact) algorithms.
The exploration of such trade-off may be useful in both 
theory and practice.
%


%% file: acknowledgment.tex
\section*{Acknowledgment}

We thank Benny Applebaum for sharing with us his result~\cite{App17}.
Pasin would like to thank Igor Shinkar and Daniel Reichman for discussions on a related problem that inspires part of the proof, and Aviad Rubinstein for useful discussions regarding FPT inapproximability. We also thank Igor for pointing us to~\cite{Kayal2014}.
Danupon and Parinya would like to thank Per Austrin for discussions. 
Parinya would also like to thank Nikhil Bansal, Jesper Nederlof, Karl Bringmann and Holger Dell for insightful discussions. 
Bundit would like to thank Uriel Feige for useful discussions on \Clique.

Marek Cygan is supported by the European Research Council (ERC) under the European Union's Horizon 2020 research and innovation programme (grant agreement No 677651). Guy Kortsarz is supported in part by NSF grants 1218620 and 1540547. Bundit Laekhanukit is partially supported by ISF Grant No. 621/12 and I-CORE Grant No. 4/11. Danupon Nanongkai is supported by the European Research Council (ERC) under the European Union's Horizon 2020 research and innovation programme under grant agreement No 715672 and the Swedish Research Council (Reg. No. 2015-04659). Pasin Manurangsi and Luca Trevisan are supported by NSF Grants No. CCF 1540685 and  CCF 1655215.

%% file: gapvapprox.tex
\section{Gap Problems vs Approximation Algorithms}
\label{app:gapvapprox}

In this section, we establish the connections between gap problems and FPT approximation algorithm by proving Proposition~\ref{prop:gapvapprox-min} and Proposition~\ref{prop:gapvapprox-max}. Proposition~\ref{prop:gapvapprox-min} is in fact implied by a result due to Chen et al.~\cite[Proposition 4]{ChenGG06}; we provide a proof of it here for completeness. On the other hand, we are not aware of any prior proof of Proposition~\ref{prop:gapvapprox-max}.

\begin{proof}[Proof of Proposition~\ref{prop:gapvapprox-min}]
We will prove the contrapositive of the statement in the proposition. Suppose that (2) is false, i.e., there exist computable functions $t: \N \to \N, f: \N \to [1, \infty)$ and an algorithm $\B$ such that, for every instance $I$ of $\Pi$, $\B$ on input $I$ runs in time $t(\opt_\Pi(I)) \cdot |I|^D$ for some constant $D$ and outputs $y \in \sol_\Pi(I)$ of cost at most $\opt_\Pi(I) \cdot f(\opt_\Pi(I))$.

Let $t': \N \to \N$ and $f': \N \to [1, \infty)$ be functions that are defined by $t'(k) = \max_{i = 1, \dots, k} t(i)$ and $f'(k) = \max_{i = 1, \dots, k} f(i)$. Since $t$ and $f$ are computable, $t'$ and $f'$ are also computable. 

Let $\A$ be an algorithm that takes in an instance $I$ of $\Pi$ and a positive integer $k$, and works as follows. $\A$ simulates an execution of $\B$ on $I$ step-by-step. If $\B(I)$ does not finish within $t'(k) \cdot |I|^{D}$ time steps, then $\A$ terminates the execution and returns 0. Otherwise, let $y$ be the output of $\B(I)$. $\A$ computes $\cost_\Pi(I, y)$; $\A$ then returns 1 if this cost is at most $k \cdot f'(k)$ and returns 0 otherwise. 

We claim that $\A$ is an $f'$-FPT gap approximation algorithm of $\Pi$. To see that this is the case, first notice that the running time of $\A$ is $O(t'(k) \cdot |I|^{D} + |I|^{O(1)})$ where $|I|^{O(1)}$ denotes the time used to compute the solution cost. Moreover, if $\opt_\Pi(I) > k \cdot f'(k)$, then it is obvious to see that $\A$ always output 0. Finally, if $\opt_\Pi(I) \leq k$, then, by our assumption on $\B$ and the definitions of $t'$ and $f'$, $\B(I)$ finishes in time $t(\opt_\Pi(I)) \cdot |I|^D \leq t'(k) \cdot |I|^D$ and the output solution $y$ has cost at most  $\opt_\Pi(I) \cdot f(\opt_\Pi(I)) \leq k \cdot f'(k)$. Hence, $\A$ always outputs 1 in this case.

As a result, $\A$ is an $f'$-FPT gap approximation algorithm for $\Pi$, which concludes our proof.
\end{proof}

\begin{proof}[Proof of Proposition~\ref{prop:gapvapprox-max}]
We will again prove the contrapositive of the statement in the proposition. Suppose that (2) is false, i.e., there exist computable functions $t: \N \to \N, f: \N \to [1, \infty)$ such that $k / f(k)$ is non-decreasing and $\lim_{k \to \infty} k/f(k) = \infty$, and an algorithm $\B$ such that, for every instance $I$ of $\Pi$, $\B$ on input $I$ runs in time $t(\opt_\Pi(I)) \cdot |I|^D$ for some constant $D$ and outputs $y \in \sol_\Pi(I)$ of cost at least $\opt_\Pi(I) / f(\opt_\Pi(I))$.

Let $t': \N \to \N$ be a function defined by $t'(k) = \max_{i = 1, \dots, k} t(i)$; clearly, $t'$ is computable.

Let $\A$ be an algorithm that takes in an instance $I$ of $\Pi$ and a positive integer $k$, and works as follows. $\A$ simulates an execution of $\B$ on $I$ step-by-step. If $\B(I)$ does not finish within $t'(k) \cdot |I|^{D}$ time steps, then $\A$ terminates the execution and returns 1. Otherwise, let $y$ be the output of $\B(I)$. $\A$ computes $\cost_\Pi(I, y)$; $\A$ then returns 1 if this cost is at least $k / f(k)$ and returns 0 otherwise. 

We claim that $\A$ is an $f$-FPT gap approximation algorithm of $\Pi$. To see that this is the case, first notice that the running time of $\A$ is $O(t'(k) \cdot |I|^{D} + |I|^{O(1)})$ where $|I|^{O(1)}$ denotes the time used to compute the solution cost. Moreover, if $\opt_\Pi(I) < k / f'(k)$, then the running time of $\B(I)$ is at most $t(\opt_\Pi(I)) \cdot |I|^D \leq t'(k) \cdot |I|^D$, which implies that $\A$ returns 0. 

Suppose, on the other hand, that $\opt_\Pi(I) \geq k$. If $\B(I)$ finishes in time $t'(k) \cdot |I|^D$, then, from the guarantee of $\B$, it must output $y \in \sol_\Pi(I)$ with $\cost_\Pi(I, y) \geq \opt_\Pi(I)/f(\opt_\Pi(I))$, which is at least $k/f(k)$ since $k/f(k)$ is non-decreasing. Furthermore, if $\B(I)$ does not finish in the specified time, then $\A$ also returns 1 as desired.

As a result, $\A$ is an $f$-FPT gap approximation algorithm for $\Pi$, which concludes our proof.
\end{proof}

%% file: gapreduction.tex
\section{Totally FPT Inapproximable Through FPT Gap Reductions (Proof of  \Cref{prop:gapred-totallyinapprox})}
\label{app:gapreduction}


	We will only show the proof when both $\Pi_0$ and $\Pi_1$ are maximization problems. Other cases can be proved analogously and therefore omitted. 	

	We assume that (i) holds, and will show that if the ``then'' part does not hold, then (ii) also does not hold.	
	Recall from \Cref{def:FPT gap reduction} that (i) implies that there exists $C, D > 0$ such that the reduction from $\Pi_0$ (with parameters $q$ and $r$) to $\Pi_1$ takes $O_{q, r}(|I_0|^C)$ time and always output an instance $I_1$ of size at most $O_{q, r}(|I_0|^D)$ on every input instance $I_0$. Now assume that  the ``then'' part does {\em not} hold; i.e. $\Pi_1$ admits a $(t(k) |I_1|^F)$-time $h$-FPT gap approximation algorithm $\A$ for some function $h(k)= o(k)$ and constant $F$.
    We will show the following claim which says that (ii) does not hold (by \Cref{def:FPT gap approx}). 
	\begin{claim}
	There exists a function $k\geq g'(k)=\omega(1)$ and an algorithm $\B$ that takes any input instance $I_0$ of problem  $\Pi_0$ and integer $k$, and in $O_k(|I_0|^{O(1)})$ time can distinguish between $\opt_{\Pi_0}(I_0) \geq k$ and $\opt_{\Pi_0}(I_0) < g'(k)$. 	
	\end{claim}
	
	We now prove the claim by constructing an algorithm $\B$ that performs the following steps. 
	Given $I_0$ and $k$, $\B$ applies the reduction on instance $I_0$ and parameters $k$ and $r=\frac{f(k)}{h(f(k))}$.
	Denote by $I_1$ the instance of $\Pi_1$ produced by the reduction, so we have that $|I_1| = O_k(|I_0|^{O(1)})$.
	The following properties are immediate from the definitions of the FPT gap reductions (\Cref{def:FPT gap reduction}).  
	\begin{itemize}
		\item If $\opt_{\Pi_0}(I_0) \geq k$, we must have $\opt_{\Pi_1}(I_1) \geq f(k)$.

		\item Also, if $\opt_{\Pi_0}(I_0) < g'(k) := g(\frac{f(k)}{h(f(k))})$, then we will have $\opt_{\Pi_1}(I_1) < r=\frac{f(k)}{h(f(k))}$  
	\end{itemize}

	Since $\A$ is an $h$-FPT gap approximation algorithm, it can distinguish between the above two cases, i.e. running $\A$ on $(I_1, f(k))$ will distinguish the above cases, therefore distinguishing between $\opt_{\Pi_0}(I_0) \geq k$ and $\opt_{\Pi_0}(I_0) < g'(k) = g(\frac{f(k)}{h(f(k))})$. 
	This algorithm runs in time $O_k(|I_1|^F) = O_k(|I_0|^{DF})=O_k(|I_0|^{O(1)})$.  
	Notice also that 
	\[g'(k) = g(\frac{f(k)}{h(f(k))}) \leq g(f(k)) \leq k \]
	where the first inequality is because  $f(k)/h(f(k)) \leq f(k)$ (recall that $h(f(k))\geq 1$ by \Cref{def:FPT gap approx}) and because $g$ is non-decreasing, and the second inquality is by the claim below.
	\begin{claim}
		For any totally-FPT-inapproximable problem $\Pi_0$, any functions $g$ and $f$ that satisfy conditions in \Cref{def:FPT gap reduction} and any integer $x$,  $g(f(x)) \leq x$.
	\end{claim}
	\begin{proof}
		For any integer $x$, consider instance $I_0$ such that $\opt_{\Pi_0}(I_0) \geq x$ (such $I_0$ exists because $\opt_{\Pi_0}=\omega(1)$; otherwise $\Pi_0$ is not totally-FPT-inapproximable (e.g. we can always output $1$ if $\Pi_0$ is a maximization problem)). By the second condition in  \Cref{def:FPT gap reduction}, $\opt_{\Pi_1}(I_1)\geq f(x)$. Applying the contrapositive of the third condition with $r=f(x)$ (thus $\opt_{\Pi_1}(I_1)\geq r$), we have $\opt_{\Pi_0}(I_0)\geq g(r)=g(f(x))$. Thus, $x\geq \opt_{\Pi_0}(I_0)\geq g(f(x))$ as claimed. 	
	\end{proof}
	To complete the proof, one only needs to argue that $g(\frac{f(k)}{h(f(k))}) = \omega(1)$, and this simply follows from the fact that $f(k) = \omega(1)$, $g(k) = \omega(1)$ and that $k/h(k) = \omega(1)$.


%
%
%
%
%


%% file: w1-inapprox.tex
\section{FTP-Inapproximability under W[1]-Hardness}
\label{apx:w1-hardness}

In this section, we show an example of problems that have no
FPT-approximation algorithm unless W[1]=FPT, namely 
the {\em maximum induced path} problem (\IPath).

In \IPath, we are given a graph $G$, and the goal is to find
a maximum size subset of vertices $S\subseteq V(G)$ such that
$S$ induces a path in $G$.
We will show that \IPath has no FPT-approximation algorithm.
Implicit in our reduction is a reduction from 
$k$-\Clique to the {\em multi-colored clique} problem.

\begin{theorem}
\label{thm:hardness-ipath}
Unless \mbox{W[1]=FPT}, for any positive integers $q:1\leq q \leq n^{1-\delta}$,
for any $\delta<0$, 
given a graph $G$ on $n$ vertices and for any function
$t:\reals\rightarrow\reals$, there is no $t(k)\poly(n)$-time
algorithm that distinguishes between the following two cases:
\begin{itemize}
\item $\IPath(G) \geq 2q\cdot k$.
\item $\IPath(G) \leq 4(k-1)$.
\end{itemize}
\end{theorem}
\begin{proof}
The reduction is as follows.
Take a graph $H$ of a $k$-\Clique instance. 
Then we construct a graph $G$ as follows.
First, we create intermediate graphs $Z_1,\ldots,Z_q$.
Each graph $Z_i$ for $i\in[q]$ is created by
making $k$ copies of $V(H)$, namely, $V_{i,1},\ldots,V_{i,k}$
and form a clique on $V_{i,j}$ for each $j\in[k]$.
So, now, we have $k$ disjoint cliques.
For each vertex $v\in V(H)$, we pick a copy of $v$,
one from each $V_{i,j}$, say $v_{i,j}$, and 
we form a clique on $\{v_{i,1},\ldots,v_{i,k}\}$.
Next, for each edge $uv\not\in E(H)$, we add edges 
$u_{i,j}v_{i,j'}$, for all $j,j'\in[k]$, where $u_{i,j}$ and $v_{i,j'}$
are the copy of $u$ in $V_{i,j}$ and the copy of $v$ in $V_{i,j'}$, 
respectively.
Next, we add a dummy vertex $x_{i,j}$ for each $V_{i,j}$
and add edges joining $x_{i,j}$ to every vertex of $V_{i,j}$ and
to every vertex of $V_{i,j-1}$ if $j \geq 2$.
Finally, we join the graph $Z_i$ for all $i\in[q]$ 
to be of a form $(Z_1,Z_2,\ldots,Z_{k})$.
To be precise, for each graph $Z_i$ with $i \geq 2$,
we join the vertex $x_{i,1}$ (which belongs to $Z_i$) 
to every vertex of $V_{i-1,q}$ (which belongs to $Z_{i+1}$).

\medskip

\noindent{\bf Completeness.}
First, suppose $\Clique(H) \geq k$.
We will show that $\IPath(G) \geq 2q\cdot k$.
We take a subset of vertices $S\subseteq V(H)$ that induces
a clique on $H$.
Let us name vertices in $S$ by $v^1,\ldots,v^k$.
For each $j\in [k]$, we pick the copies $v^j_{i,j}$ of
$v^j$ from $V_{i,j}$ for all $i\in[q]$.
We then pick all the vertices $x_{i,j}$ for $i\in[k]$ and $j\in[q]$.
We denote this set of vertices by $S'$. 
It is not hard to see that for any distinct vertices $v^{j},v^{j'}\in S$,
their copies $v^j_{i,j}$ and $v^{j'}_{i',j'}$ are not adjacent,
and each vertex $x_{i,j}$ has exactly two neighbors:
$v^j_{i,j}$ and $u^{j-1}_{i,j-1}$ (or $u^k_{i-1,k}$).
Therefore, $S'$ induces a path in $G$ of size $2qk$. 

\medskip

\noindent{\bf Soundness.}
Suppose $\Clique(H) < k$, i.e., $H$ has no clique of size $k$.
We will show that $\IPath(G)\leq 4(k-1)$.
To see this, let $S'\subseteq V(G)$ be a subset of vertices that
induces a path $G[S']$ in $G$.
Observe that, for $i\in[q]$, $G[S']\cap Z_i$ must be a path 
of the form $(x_{i,a},v^a_{i,a},\ldots,x_{i,k},v^b_{i,b})$.
Moreover, $v^\ell_{i,\ell}$ and $v^{\ell'}_{i,\ell'}$ are not adjacent in $G$ 
for any $\ell\neq\ell'$, meaning that $v^\ell_{i,\ell}$ and
$v^{\ell'}_{i,\ell'}$ are not copies of the same vertex in $H$,
and the set $\{v^{\ell}\}_{a\leq \ell\leq b}$ induces a clique in $H$.
Thus, $a-b+1 < k$, and $G[S']\cap Z_i$ can have at most $2(k-1)$
vertices.
It follows that any induced path $G[S']$ of $G$ can contain vertices
from at most two subgraphs, say $Z_i$ and $Z_{i+1}$.
Therefore, we conclude that $|S'| \leq 4(k-1)$.
\end{proof}

The FPT-inapproximable of \IPath follows directly from
Theorem~\ref{thm:hardness-ipath}.

\begin{corollary}
\label{cor:fpt-inapprox-ipath}
Unless W[1]=FPT, there is no $f(k)$-approximation algorithm for
\IPath that runs in $t(k)\poly(n)$-time for any functions $f$ and
$t$ depending only on $k$. 
\end{corollary}

%% file: appendix.tex

\section{Known Connections between Problems}
\label{app:trivial-eq}
In this section, we discuss known equivalences between problems in more detail. 

\paragraph{Dominating Set and Set Cover:}
It is easy to see that \DomSet is a special case of \SetCov,
and the reduction from \SetCov to \DomSet is by phrasing
$\mathcal{U}$ and $\mathcal{S}$ as vertices,
forming a clique on $\mathcal{S}$ and there is an edge joining 
a subset $S_i \in \mathcal{S}$ and element $u_j \in \mathcal{U}$ if and only if $u_j$ is an element in $S_i$.

\paragraph{Induced Matching and Independent Set:} We show that Induced Matching is at least as hard to approximate as Independent Set. 
Let $G$ be an input graph of Independent Set. 
We create a graph $G'$ by, for each vertex $v \in V(G)$, create a vertex $v'$ and an edge $v v'$. 
Notice that any independent set $S$ of $G$ corresponds to an induced matching in $G'$: For each $v \in S$, we have an edge $v v'$ in the set $\mathcal{M}$. 
Conversely, for any induced matching $\mathcal{M}$ of $G'$, we may assume that the matching only chooses edges of the form $vv'$. 

\paragraph{More hereditary properties:} We discuss some more natural problems in this class. 
If we define $\Pi$ to be a set of all planar graphs, this is hereditary. 
The corresponding optimization problem is that of computing a maximum induced planar graphs. 
If we define $\Pi$ to be a set of all forests, this is also hereditary, and it gives the problem of computing a maximum induced forest.

\section{Proof Sketch of \Cref{thm:Gap-ETH restated}}
\label{sec:restate-gap-eth}

We will sketch the proof of \Cref{thm:Gap-ETH restated}.

In the forward direction, we use a standard reduction, which is sometimes referred to as the {\em clause-variable game}~\cite{AIM14}. Specifically, we transform a $3$-\SAT instance $\psi$ on $n$ variables $x_1,\ldots,x_n$ and $m$ clauses $C_1,\ldots,C_m$ into a label cover instance $\Gamma=(G=(U,V,E),\Sigma_U,\Sigma_V,\Pi)$ by transforming clauses into left vertices in $U$ and variables into right vertices in $V$,
and there is an edge joining a pair of vertices $C_i$ and $x_j$ if
$x_j$ appears in $C_i$.
We take partial assignments as the label sets $\Sigma_U$ and $\Sigma_V$,
and a constraint on each edge asks for a pair $(\alpha,\beta)$ of
labels that are  {\em consistent}, i.e.,
they assign the same value to the same variable (e.g.,
$\alpha=(x_1:1,x_2:0,x_3:1)$ and $\beta=(x_1:1)$  are consistent
whereas $\alpha$ is not consistent with $\beta'=(x_2:1)$),
and $\alpha$ causes $C_i$ to evaluate to true
(i.e., some of the literal in $C_i$ is assigned to true by $\alpha$).
We denote the evaluation of a clause $C_i$ on a partial assignment
$\alpha$ by $C_i(\alpha)$.

To be precise, we have
\[
\begin{array}{l}
 U = \{C_1,\ldots,C_{m}\}, \quad
 V = \{x_1,\ldots,x_{n}\}, \quad
 E = \{C_ix_j: \mbox{$x_j$ appears in the clause $C_i$}\}\\
\Sigma_U = \{0,1\}^3, \quad
\Sigma_V = \{0,1\}, \quad
\Pi_{C_ix_j} = \{(\alpha,\beta): \mbox{$\alpha$ and $\beta$ are consistent
               $\land$ $C_i(\alpha)=\mbox{true}$}\}
\end{array}
\]

It can be seen that $\MaxCov(\Gamma)=\SAT(\psi)$ since the only way
to cover each node $C_i\in U$ is to pick assignments to all vertices
adjacent to $C_i$ so that they are all consistent with the assignment
$\alpha=\sigma_V(C_i)$ (and that $C_i(\alpha)=\mbox{true}$).

The converse direction is not straightforward. 
We apply H{\aa}stad \cite{Hastad01} reduction\footnote{Here we apply only the Hastad's reduction from label cover to 3\SAT, without parallel repetition.} to reduce an instance $\Gamma$ of \MaxCov
to a $3$-SAT instance of size 
$f(|\Sigma_U|+|\Sigma_V|)\cdot \mbox{linear}(|U|+|V|)$
with a hardness gap $1-\varepsilon$, for some constant
$\varepsilon>0$ 
(the hardness gap is different from the original \MaxCov instance).
Note that $f$ in the H{\aa}stad's construction is a doubly exponential
function.
The equivalent between \MaxCov and $3$-\SAT holds only when
$|\Sigma_U|+|\Sigma_V|$ is constant (or at most $\log\log (|V|+|U|)$).
\qed

%% file: gap-eth.tex
\section{On Gap-ETH} \label{app:gap-eth}

While Gap-ETH may sound like a very strong assumption, as pointed out in~\cite{Dinur16, ManR16}, there are a few evidences suggesting that the conjecture may indeed be true:
\begin{itemize}
\item In a simplified and slightly inaccurate manner, the PCP theorem~\cite{AroraS98,AroraLMSS98} can be viewed as a polynomial time reduction that takes in a 3-CNF formula $\Phi$ and produces another 3-CNF formula $\Phi'$ such that, if $\Phi$ is satisfiable, then $\Phi'$ is satisfiable, and, if $\Phi$ is unsatisfiable, $\Phi'$ is not only unsatisfiable but also not even $0.99$-satisfiable. By now, it is know that the size of $\Phi'$ can be made size small as $n \polylog(n)$ where $n$ is the size of $\Phi$~\cite{Dinur07}. This means that, assuming ETH, Gap-3\SAT cannot be solved in $2^{o(n/\polylog n)}$ time, which is only a factor of $\polylog n$ off from what we need in Gap-ETH. Indeed, as stated earlier, if a linear-size PCP, one in which $\Phi'$ is of size linear in $n$, exists then Gap-ETH would follow from ETH.
\item No subexponential-time algorithm is known even for the following (easier) problem, which is sometimes referred to as \emph{refutation of random 3-\SAT}: for a constant density parameter $\Delta$, given a 3-CNF formula $\Phi$ with $n$ variables and $m = \Delta n$ clauses, devise an algorithm that outputs either SAT or UNSAT such that the following two conditions are satisfied:
\begin{itemize}
\item If $\Phi$ is satisfiable, the algorithm always output SAT.
\item Over all possible 3-CNF formulae $\Phi$ with $n$ clauses and $m$ variables, the algorihtm outputs UNSAT on at least 0.5 fraction of them.
\end{itemize} 

Note here that, when $\Delta$ is a sufficiently large constant (say 1000), a random 3-CNF formula is, with high probability, not only unsatisfiable but also not even $0.9$-satisfiable. Hence, if Gap-ETH fails, then the algorithm that refutes Gap-ETH will also be a subexponential time algorithm for refutation of random 3-\SAT with density $\Delta$.

Refutation of random 3-\SAT, and more generally random CSPs, is an important question that has connections to many other fields, including hardness of approximation, proof complexity, cryptography and learning theory. We refer the reader to~\cite{AOW15} for a more comprehensive review of known results about the problem and its applications in various areas. Despite being intensely studied for almost three decades, no subexponential-time algorithm is known for the above regime of parameter. In fact, it is known that the Sum-of-Squares hierarchies cannot refute random 3-\SAT with constant density in subexponential time~\cite{Gri01,Sch08}. Given how powerful SDP~\cite{Rag08}, and more specifically Sum-of-Squares~\cite{LRS15}, are for solving (and approximating) CSPs, this suggests that refutation of random 3-\SAT with constant density, and hence Gap-3\SAT, may indeed be exponentially hard or, at the very least, beyond our current techniques.

\item Dinur speculated that Gap-ETH might follow as a consequence of some cryptographic assumption~\cite{Dinur16}. This was recently confirmed by Applebaum~\cite{App17} who showed that Gap-ETH follows from an existence of any exponentially-hard locally-computable one-way function. In fact, he proved an even stronger result that Gap-ETH follows from ETH for some CSPs that satisfy certain ``smoothness'' properties.
\end{itemize}

Lastly, we note that the assumption $m = O(n)$ made in the conjecture can be made without loss of generality. As pointed out in both~\cite{Dinur16} and~\cite{ManR16}, this follows from the fact that, given a 3-\SAT formula $\phi$ with $m$ clauses and $n$ variables, if we create another 3-\SAT formula $\phi'$ by randomly selected $m' = \Delta n$ clauses, then, with high probability, $|\SAT(\phi)/m - \SAT(\phi')/m'| \leq O(1/\Delta)$.